\documentclass[12pt, oneside]{article}
\usepackage{times}
\usepackage{filecontents}
\usepackage{comment}
\usepackage[moderate]{savetrees}
\usepackage{geometry}
\usepackage[dvipsnames]{xcolor}
\usepackage{graphicx}
\usepackage{setspace}
\usepackage{calrsfs}
\usepackage{amsmath}
\usepackage{amsfonts}
\usepackage{amssymb}
\usepackage{amsthm}
\usepackage{epstopdf,subcaption,ushort}

\usepackage[round]{natbib}

\usepackage{enumerate}

\usepackage[shortlabels]{enumitem}
\usepackage[colorlinks=true,linkcolor=blue]{hyperref}
\hypersetup{colorlinks=true,citecolor=blue}
\usepackage{cleveref}
\usepackage{footnotebackref}

\usepackage{bbm}
\usepackage{ulem}
\newtheorem{theorem}{Theorem}
\newtheorem{lemma}{Lemma}
\newtheorem{proposition}{Proposition}
\newtheorem{corollary}{Corollary}

\newtheorem{claim}{Claim}
\newtheorem{observation}{Observation}

\newtheorem{remark}{Remark}
\newtheorem{definition}{Definition}

\newtheorem{assumption}{Assumption}

\mathchardef\ordinarycolon\mathcode`\:
\mathcode`\:=\string"8000
\begingroup \catcode`\:=\active
\gdef:{\mathrel{\mathop\ordinarycolon}}
\endgroup

\DeclareMathOperator*{\argmin}{arg\,min}
\DeclareMathOperator*{\argmax}{arg\,max}

\begin{document}

\title{\textbf{Turning the Ratchet: \\Dynamic Screening with Multiple Agents}\footnote{We thank Ralph Boleslavsky, Rahul Deb, Laura Doval, 
Tom Gresik, Moritz Meyer-ter-Vehn, Alessandro Pavan, Alexandre Poirier, Michael Grubb, Siyang Xiong, 
Yuzhe Zhang, and participants at various seminars and conferences for helpful discussions and comments. All errors are our own. \textit{Email:} Ekmekci (mehmet.ekmekci@bc.edu); Maestri (lucas.maestri@fgv.br); Wei (dwei10@ucsc.edu).}}
\author{%
\begin{tabular}{ccccc}
\textsc{Mehmet Ekmekci
} & \textsc{Lucas Maestri
} & \textsc{Dong Wei
}\\
{ \textit{Boston College}} & { \textit{FGV EPGE}} & { \textit{UC Santa Cruz}}%
\end{tabular}
}
\date{March 2025
}
\maketitle

\begin{abstract}
We study a dynamic contracting problem with multiple agents and limited commitment. A principal seeks to screen efficient agents using one-period contracts, but is tempted to revise contract terms upon knowing an agent's type. Alterations of contracts are observable and, hence, whenever past promises are broken future information revelation stops. We provide necessary and sufficient conditions under which information revelation can be fostered. For sufficiently patient players, private information is either never revealed or fully revealed in a sequential manner. Optimal contracts provide high-powered incentives upon initial disclosure of an agent’s type, and rewards for information revelation vanish over time.
 \bigskip
 
\noindent \textit{Keywords:} ratchet effect, dynamic contracts, screening, limited commitment\\
\noindent \textit{JEL classification:} C73, D82, D86, J31  
\end{abstract}

\newpage
\clearpage

\section{Introduction}
In dynamic relationships with private information, information revelation is often critical for improving allocative efficiency and welfare. For example, in a worker-manager relationship, the worker's ability is usually his private information; obtaining this information would allow the manager to assign different tasks based on the worker's ability, which is beneficial for social surplus and the firm's profit. Nevertheless, the revelation of private information can be difficult to incentivize, when the uninformed party lacks a commitment to not exploiting such private information. This causes the well-known ``ratchet effect." Indeed, if a worker expects his manager to implement worse terms of trade (such as assigning him more work or reducing his wage rate) after the manager knows his capability, the worker would be reluctant to reveal such information in the first place. {
This concern is especially pronounced when the firm holds some monopsony power and the worker cannot credibly threaten to leave.}

In various organizations, however, a person in a supervising position is often managing many subordinates at the same time. In such scenarios where one principal engages with a large population of agents, the principal's commitment power may be partially restored, because if she breaks her promise with one of the agents, there can be potential punishments imposed on her by other agents; the linkage between these relationships could then push the principal to be more restrained in each of them. In this paper, we study the effect of a principal simultaneously contracting with multiple agents, and whether and how that can alleviate the ratchet-effect problem. In particular, we are interested in addressing the following questions. First, how should dynamic incentives be designed in order to incentivize information revelation over time? Relatedly, what are the qualitative features of the optimal contract? Moreover, what determines whether any information is revealed in equilibrium?

To answer these questions, we study a stylized model in which a principal (she) interacts with the same mass of agents (he) over time. An agent can be either the efficient (i.e., low-cost) type or the inefficient (i.e., high-cost) type, which is his private information. In each period, an agent produces an observable quantity, incurs a production cost, and collects some payment from the principal. Meanwhile, in each period the principal tries to induce a distinct subset of efficient agents to reveal their type. She does so by promising a sequence of separating contracts to these agents while also offering a pooling contract to those who wish to remain unrevealed. Free to make any promise about the future, the principal can only commit to one-period contracts for the current period, and her promise about what contract to offer in any future period can be reconsidered and reneged on when the day arrives. Nonetheless, reneging is not costless. We assume that the principal's reneging is publicly observable to \textit{all} agents, and once it occurs, those agents who haven't yet revealed their private information will stop revealing altogether. In this model, the principal's fear for losing future information then becomes a deterrent to the excessive exploitation of a low-cost agent after his type is revealed. In turn, this gives each agent more incentives to reveal his type in the first place, which potentially fosters information transmission and alleviates the ratchet effect.

We say that a mechanism is revealing if it successfully induces some efficient agents to reveal their type over time. Depending on the model parameters, an optimal mechanism can be revealing or non-revealing.  Our first main result characterizes the qualitative features of an optimal mechanism when it is revealing.\footnote{When an optimal mechanism is non-revealing, it is simply a static pooling contract that extracts the full surplus from the inefficient type and leaves some rent to the efficient type.} Our second main result provides a necessary and sufficient condition under which every optimal mechanism is revealing when players are sufficiently patient.

When information revelation is feasible, in any optimal mechanism, every period there is a distinct subset of efficient agents choosing to reveal their type, and all efficient agents will reveal at some point. In other words, information revelation is \textit{gradual} and \textit{eventual.} Gradual revelation ensures that looking into the future from each period's perspective, there is always new information to arrive later. This gives the principal sustained incentives to continue paying the promised information rent, instead of fully exploiting every low-cost agent who has revealed his type. Eventual revelation is implied by optimality, for otherwise the principal can make a higher profit in another mechanism with more agents revealing in every period.

The pooling contract offered in each period (taken by all inefficient agents and those efficient ones who wish to remain unrevealed) exhibits a simple and intuitive dynamics. Specifically, each pooling contract asks the agent to produce a quantity less than the inefficient type's surplus-maximizing level and remunerates him at the high cost of production;\footnote{Downward distortions of pooling quantities are typical in standard screening models (i.e., an inefficient type is never overworked). They are used to reduce the information rent of the efficient type.} as the efficient agents gradually drop out of the pool, the demanded quantity converges to the surplus-maximizing level over time.

On the other hand, once an efficient agent reveals his type, he will be getting the separating contracts starting from the current period. This sequence of contracts consists of two phases.\footnote{The sequence of separating contracts an agent gets can differ depending on when he reveals his efficiency. Nevertheless, these sequences of contracts all share the same qualitative dynamics.} In the first several periods after revelation, which we call the \textit{``reward phase,"} the agent is asked to produce greater than the efficient type's surplus-maximizing quantity and receives a payment equal to the inefficient type's high cost. During these periods, an efficient agent, while working intensely, will be collecting positive rents. The demanded quantity and offered payment in the contract will gradually decline during the reward phase until the next phase, called the \textit{``exploitation phase,"} is reached. In the exploitation phase, the principal extracts the entire surplus from an efficient agent by asking him to produce at his surplus-maximizing level and paying him at the low cost of production.

The dynamics of the separating contracts are shaped by the interplay between the principal's limited commitment and the agent's incentive compatibility. First, the principal would like to frontload an efficient agent's information rent as much as possible. This is because by paying more today (with commitment) and promising less in the future (without commitment), the principal can deliver the same value to an efficient agent while relaxing her \textit{non-reneging} constraint.\footnote{In our model, although the principal could renege on her promise when a future date comes, any \textit{implementable} allocation requires the principal to find it optimal to always keep her promise. Formally, this is captured by a ``non-reneging" constraint in each period.} However, paying all information rent to an agent in one period can be infeasible, because such a high upfront payment may induce the inefficient agent to pretend to be efficient, reap the high payment, and then exit from the relationship. To hedge against this ``take-the-money-and-run" strategy by the inefficient type, the principal has to split the total information rent to an efficient agent into the \textit{first several periods} after his type is revealed, making the reward phase nontrivial. Moreover, the upward distortion of the low-cost type's quantity is also driven by the principal's incentive to frontload information rent. Indeed, if his production is at or below the surplus-maximizing level, the principal can increase the quantity and perturb the payment to further relax her non-reneging constraint and/or increase the surplus.


Next, we derive a necessary and sufficient condition under which the agent's private information can be elicited. With limited commitment, the main challenge in encouraging information revelation is to resolve the agent's concern---that the principal would not honor the promised information rent after knowing his type. Naturally, our easy-to-check condition on the model primitives involves a comparison between the principal's benefit and cost of keeping promises. From the principal's perspective, the benefit of honoring all promises (for one more period) is that more efficient agents will reveal their type, and this additional information can be utilized later. Meanwhile, the cost of keeping promises (for one more period) comes from two sources: \textit{i)} payment of information rent to those agents who are still in the reward phase of their separating contracts, and \textit{ii)} surplus loss due to the quantity distortions in these contracts. We show that, when the players are sufficiently patient, every optimal mechanism is revealing if and only if the benefit of keeping promises outweighs its cost. This necessary and sufficient condition uncovers a novel tradeoff between eliciting future information and exploiting past disclosures; it also enables us to examine the effect of skill heterogeneity on information revelation.

Our analysis highlights the key distinctions between contracting with one agent and that with many agents. In the one-agent case, if the principal breaks her promise, there is little punishment this single agent can impose on the principal (other than perhaps quitting the relationship). In the many-agent case, different agents could reveal their private information at different points in time, and the communication among these agents would publicize any mistreatment of an agent by the principal to all other agents. Consequently, those who have not revealed their type could then impose a collective punishment by stopping revealing information altogether. This much harsher punishment induces the principal to honor the promised information rent, which in turn gives the efficient agents more incentives to reveal their type. This result suggests that, in organizations with commitment issues, enhancing ties and facilitating (informal) communication between employees could help overcome the ratchet effect and foster information transmission.

{
Our characterization of optimal mechanisms also generates sharp predictions about the life cycle of a worker's contract in labor markets with a significant lack of commitment. For an incompetent worker, his wage and work amount remain at a relatively low level throughout his career. More interestingly, a competent worker's contract exhibits an inverse U-shaped pattern: His wage and work amount may start low, but once the principal learns about his competence, he will receive a major promotion that comes with a significant increase in both wage and responsibilities, after which they are  faded down over time as the principal begins to make full use of her knowledge about the worker's type (see Figure \ref{fig:lifecycle}). In industries where contracts are often renegotiatied and skills are project-specific, such as IT consulting and entertainment, we find examples that seem consistent with our predictions. We discuss them in detail at the end of Section \ref{sec:structure}. 
}

Apart from contracting in labor markets, our model and insights extend naturally to other related settings. For example, large retail corporations, such as Walmart, engage with thousands of suppliers who possess valuable but private information about production costs, etc. The ratchet effect discourages these suppliers from revealing their true costs to the retailer, as such disclosures might trigger worse terms of trade and erode their profits. Our analysis offers a recipe to such retailers for leveraging the interdependence between multiple relationships, fostering information transmission, and improving profitability.
\medskip

\noindent \textbf{Related Literature.} This paper contributes to the literature on ratchet effect. The seminal works of \cite{freixas1985} and \cite{laffont1988,laffont1993} first model the ratchet effect and show how it severely impairs information revelation in
dynamic relationships. \cite{gerardi2020} extend the model to infinite horizon and find that the unique equilibrium outcome in the patient limit is pooling, so long as the probability of the inefficient type is not too low. 

The literature has identified several institutions that help mitigate the ratchet effect. \cite{kanemoto1992} point out that competition for second-hand workers can help firms with their commitment problems; \cite{carmichael2000} find that a patient firm may have an incentive to honor its promises to short-lived workers; when there are productivity shocks, \cite{avidit2017} show that the principal can progressively learn the agent's private information. {
By contrast, even with a monopsonistic firm and persistent private information, our paper demonstrates how the interdependence between multiple relationships with long-lived workers can restore the firm's commitment power and relieve the ratchet effect.} Moreover, whenever information revelation is feasible, \textit{all} private information will be revealed eventually.

This paper also belongs to the broader literature on screening with limited commitment. In the context of dynamic pricing, \cite{hart1988} analyze a model in which a principal rents over time an indivisible good to an agent who is privately informed about his valuation; \cite{skreta2006,skreta2015} and \cite{liu2019} use a mechanism-design approach to study optimal auctions with limited commitment; \cite{bonatti2011} studies optimal dynamic menu pricing when past purchases can affect a consumer's learning of product quality and thus his future demands; \cite{deb2015} demonstrate how delaying contracting with high-valuation buyers can help a seller credibly maintain higher prices in the future. More generally, \cite{strulovici2017} and \cite{maestri2017} examine screening models with renegotiation opportunities; \cite{bester2001} and \cite{doval2022} develop revelation principles for models involving limited commitment and a single agent. Our paper differs from these studies by analyzing the interaction between a principal and a continuum of agents, and allowing for allocations that depend not only on each agent's type but also on his dynamic choices.

There is also a macroeconomics literature on optimal government policies with a lack of commitment and a continuum of privately
informed agents \citep[see, for example,][]{sleet2008,albanesi2006,acemoglu2010,scheuer2016}. These papers all assume that agents of the same type receive the same allocation, while the principal in our model can personalize the allocation to (ex-ante identical) agents who reveal their private information at different points in time. {
\cite{golosov2021} analyze optimal unemployment insurance with agents who cannot falsely claim to have job offers, while the current model assumes fully soft information, allowing agents to misreport freely. The different nature of private information in our model generates binding take-the-money-and-run constraints and novel predictions about the equilibrium dynamics.}

Our paper is reminiscent of other contributions highlighting that interacting with multiple agents help the principal keep her promises. In a moral-hazard setting with complete information, \cite{levin2002} shows that through pooling incentives the principal does better by making promises to several agents than to only one. This differs from our model as there is no scope for information revelation. \cite{ausubel1989} find that when selling a durable good to a continuum of consumers, a patient seller can obtain a payoff arbitrarily close to the monopoly profit in an equilibrium in which she slowly decreases the price over time. In our model, the key tradeoff faced by the principal is between reneging on the information rent promised in the past and eliciting new information in the future. This tradeoff is only present in our paper because in a durable-good setting a buyer leaves the market immediately after he reveals his type (i.e., buys the good).

\section{Model}
\subsection{Players, Types, and Payoffs}
A principal (she) interacts with a continuum of agents (he) of measure $1$. Time is discrete and infinite, and both parties discount the future with a discount factor $\delta \in \left( 0,1\right)$. Each agent has a type $\theta \in \left\{ L,H\right\}$, which affects the agent's cost of production and is his private information. Initially, the measure of the \textit{low(-cost)} type is $\alpha_0\in (0,1)$ and the measure of the \textit{high(-cost)} type is $1-\alpha_0$.  

In period $t$, if a type-$\theta$ agent produces quantity $q_t$ and receives a payment of $x_t$, his payoff is%
\[
x_{t}-C_{\theta} \left(q_{t}\right), 
\]%
where $C_{\theta}(q_t)$ is the cost of production. We assume that the quantity belongs to some compact set $[0,\Bar{q}]$, and that $C_{\theta}(\cdot)$ is smooth, strictly
increasing, and convex, satisfying $C_{\theta} \left( 0\right) =0$ for
every $\theta$ and $C'_{H}\left(q\right) >C'_{L}\left(q\right)$ for every $q\in  [0,\bar{q}]$. This implies that $\Delta C(\cdot) := C_H(\cdot) - C_L (\cdot)$ is nonnegative and strictly increasing. We further assume that $\Delta C$ is concave, i.e., $\Delta C''(\cdot) \leq 0$.\footnote{The concavity of $\Delta C$ ensures a cleaner statement of our main results. Our results still hold qualitatively even without this assumption. See Section \ref{sec:concavedeltac} for a discussion.} These assumptions are satisfied if, for example, $C_\theta(q) = \theta q$.

The principal's payoff from an agent is%
\[
v(q_{t})-x_{t}, 
\]%
where $v(\cdot )$ is strictly increasing, strictly concave, and
differentiable, satisfying $v(0) = 0$.

In any period, the social surplus from type $\theta$ producing quantity $q$ is%
\[
\pi _{\theta}(q):=v(q)-C_\theta(q). 
\]
For each $\theta\in \{L,H\}$, we define the surplus-maximizing quantity
\[
q_{\theta}^{\ast }:=\argmax_{q}\ \pi _{\theta}(q). 
\]%
We assume the solutions are interior, so that $0 < \pi _{H}(q_{H}^{\ast })<\pi _{L}(q_{L}^{\ast })$ and $0 < q_H^* < q_{L}^* < \Bar{q}$. 

If at any point an agent decides to irreversibly drop out of the relationship, he will receive a payoff of $0$ as his outside option.

Throughout this paper, we maintain the following assumption that the initial measure of the low-type agents is not too large.
\begin{assumption}\label{ass:1}
$\alpha_0 < \frac{\pi_H(q_H^*)}{\pi_L(q_L^*)}$.
\end{assumption}
Assumption \ref{ass:1} ensures that the principal should \textit{not} simply ignore the high types and extract full surplus from and only from the low types. See Section \ref{sec:alpha} for a discussion on relaxing this assumption.

\begin{remark}
The assumption of a continuum of agents accounts for various situations where new information could trickle in and there is uncertainty about the extent of information that will need to
be disclosed in the future. In Section \ref{sec:cardinality}, we discuss to what extent our results and insights remain valid when this assumption is relaxed. 
\end{remark}




\subsection{The Contracting Environment}
We are interested in studying a contracting environment with three key features:
\begin{enumerate}
    \item The principal can only commit to one-period contracts and may renege on her past promises. However, any reneging is observable to \textit{all} agents and could be ``punished" by the agents.
    \item Each agent has private information about his type, and any information revelation is voluntary.
    \item Each agent can irreversibly leave the relationship in any period by refusing the current contract.
\end{enumerate}

To better understand the limited commitment assumption, note that in this dynamic setting, the principal can not only specify a contract (typically describing the amount of work and compensation) for the current period, but also make promises about what contracts to offer in future periods. We assume that any contract for the current period, once accepted, is binding and enforceable. On the other hand, when a future period comes, the principal could ``renege" on her promise by presenting a contract that is different from what she has promised to offer. This form of limited commitment can be caused by various factors, such as imperfect enforcement of long-term contracts, limited validity of verbal promises, and/or periodical leadership changes within an organization. 
\subsubsection*{Sequential Revelation Mechanisms}
We focus our attention on a natural class of \textbf{sequential revelation mechanisms} where low-cost agents can come forward to gradually reveal their efficiency over time. In each period, any agent can announce that his type/cost is low, and the contract each agent gets depends on whether he has made such an announcement.

Formally, at the beginning of each period $t$, the principal proposes a \textit{pooling contract} $(q_t^P,x_t^P)$ and a sequence of \textit{separating contracts} $\left\{\left(
q_{\tau}^{S}\left( t\right) ,x_{\tau}^{S}\left( t\right) \right) \right\}_{\tau\geq
t}$. If an agent has not made any announcement (at $t$ or before), he will be given the pooling contract $(q_t^P,x_t^P)$. If, in period $t$, an agent announces that his type is low, then he will be given the separating contract $\left(q_{t}^{S}\left( t\right) ,x_{t}^{S}\left( t\right)\right)$ in the current period and be promised a sequence of separating contracts  $\left\{\left(
q_{\tau}^{S}\left( t\right) ,x_{\tau}^{S}\left( t\right) \right) \right\}_{\tau>
t}$ for all future periods $\tau>t$, where the $t$ in the parentheses indicates the time of his revelation.\footnote{A contract $(q_t,x_t)$ in period $t$ should be interpreted as a take-it-or-leave-it offer: An agent will be paid $x_t$ if he produces $q_t$ (or more), but he can also leave the relationship by refusing the contract.} We allow the sequence of separating contracts to depend on the period in which the agent announces his (low) type; thus, low-type agents revealing in two different periods may be promised two distinct sequences of separating contracts.

From an agent's perspective, in a sequential revelation mechanism, his choice set in each period depends on whether he has made an announcement or not. Specifically,
\begin{itemize}
    \item For an agent who has not made any announcement before $t$, in period $t$ he can choose between: \textit{i)} continuing to take the pooling contract (i.e., remaining silent), \textit{ii)} starting to take the separating contract (i.e., revealing to be efficient), and \textit{iii)} quitting the relationship.\footnote{We interpret the agent's choice within a menu as making a report (or remaining silent) about his type.}
    \item For an agent who announced to be the low type at some $T<t$, in period $t$ he can choose between: \textit{i)} taking the separating contract $(q_t^S(T),x_t^S(T))$ and \textit{ii)} quitting the relationship.
\end{itemize}
Indeed, once an announcement is made, the agent is identified by the principal and can no longer take the pooling contract in any future period. 

From the principal's perspective, because of her limited commitment, in each period $t$ she can freely propose a contract $(\tilde{q}_t^S(T),\tilde{x}_t^S(T))$ to agents who revealed at each $T<t$. This contract can potentially differ from the one she promised to offer $({q}_t^S(T),{x}_t^S(T))$, although in equilibrium they must be the same.

In summary, the timeline within each period $t$ is as follows.
\begin{enumerate}
    \item The principal proposes a pooling contract $(q_t^P,x_t^P)$, a sequence of separating contracts $\left\{\left(
q_{\tau}^{S}\left( t\right) ,x_{\tau}^{S}\left( t\right) \right) \right\}_{\tau\geq
t}$ (for those who wish to reveal at $t$), and $(\tilde{q}_t^S(T),\tilde{x}_t^S(T))$ for each $T<t$ (for those who revealed before $t$).\footnote{One can allow the principal to post any finite menu with $K$ contracts to agents who haven't yet revealed and focus on equilibrium outcomes where at most two of these contracts are chosen within each period.}
\item[2(a).] If an agent always took the pooling contract prior to $t$, he can choose between $(q_t^P,x_t^P)$, $(q_t^S(t),x_t^S(t))$, and quitting. 
\item[2(b).] If an agent took the separating contract $(q_T^S(T),x_T^S(T))$ at some $T<t$, he can choose between $(\tilde{q}_t^S(T),\tilde{x}_t^S(T))$ and quitting.
\item[3.] The chosen contracts are enforced between the principal and each agent.
\end{enumerate}

\begin{remark}
    Sequential revelation mechanisms introduced above are suited to address our key research questions---that whether and how the interdependence between multiple relationships can facilitate information revelation over time. These mechanisms focus particularly on the revelation of the efficient agent's type, as such information is most useful for the principal and creates critical tensions between the two parties.\footnote{In our model, knowing that an agent is inefficient, compared to not knowing his type, does not change the principal's optimal contract with this agent.} Nevertheless, whether or not the restriction to this class of mechanisms is without loss remains an open question, as we are not aware of a revelation principle applicable to this setting. In particular, the revelation principle developed in \cite{doval2022} for a single agent does not extend directly to our setting with multiple (in fact, a continuum of) agents.\footnote{Indeed, with multiple agents, attempting to extend their ``canonical mechanisms" requires us to keep track of the principal's belief evolution about all agents' types, which could jointly determine the entire allocation. It is an open question whether such an extension remains canonical in multi-agent environments, and even if so, how it can be applied to characterize optimal mechanisms.}
\end{remark}

\subsection{Implementable Allocations and the Principal's Problem}
\subsubsection*{Allocations}
In general, an allocation is a collection of quantity and payment streams, $\left\{\left(q_t^i,x_t^i\right)\right\}_{t\geq 0,i\in [0,1]}$, which describes the work and compensation of each agent in every period. In a sequential revelation mechanism, because the contract each agent gets only depends on his type and the time of revelation (if his type is low), an allocation can be summarized by
\begin{equation}\label{eq:allocation}
    \left\{\left(q_t^P,x_t^P\right), \left(q^S_\tau(t),x^S_\tau(t)\right), r_t\right\}_{t\geq 0, \tau\geq t},
\end{equation}
which specifies each agent's contracts before and after revealing (that his type is low), and the measure of low-cost agents, $r_t$, who reveal in each period $t$.\footnote{For any $t$ such that $r_t = 0$, one can simply set $\left(q^S_\tau(t),x^S_\tau(t)\right) = \left(q_\tau^P,x_\tau^P\right)$ for all $\tau\geq t$, making the separating contracts starting from a non-revealing period the same as the pooling contracts. Alternatively, one can think $\left(q^S_\tau(t),x^S_\tau(t)\right)$ as defined only for $t$ such that $r_t > 0$.}$^,$\footnote{In an allocation described by \eqref{eq:allocation}, we implicitly assume that each individual agent is using a pure strategy: all high-type agents never choose to make a (false) announcement, and each low-type agent reveals with probability one in some period. The revealing period may differ across different low-type agents, thus $r_t$ can be positive in multiple periods despite the fact that there is no randomization on the individual level.} Define $R_t : = \sum_{\tau < t} r_t$ to be the mass of agents who reveal before $t$. Since this paper focuses solely on sequential revelation mechanisms, we only consider allocations that can be represented by \eqref{eq:allocation}.

Given an allocation from a sequential revelation mechanism, in period $t$, the continuation payoff of a type-$\theta$ agent if he were never to reveal is
\begin{align}
    U_{\theta,t}^P := \left( 1-\delta \right) \sum_{\tau =t}^{\mathbb{\infty }}\delta
^{\tau -t}\left[ x_{\tau}^P-C_\theta \left( q_{\tau}^P\right) \right].
\end{align}
The continuation payoff of a low-type agent who has revealed at $T\leq t$ (i.e., who starts to take the separating contract at $T$) is
\begin{align}\label{eq:ULtT}
    U_{L,t}^{S}\left( T\right) :=\left( 1-\delta \right) \sum_{\tau =t}^{%
\mathbb{\infty }}\delta ^{\tau -t}\left[ x_{\tau}^{S}(T)-C_L \left(q^S_\tau(T)\right) \right].
\end{align}
We define the continuation profit of the principal as
\begin{align}
    \Pi_t := &\left( 1-\delta \right) \sum_{\tau =t}^{\mathbb{\infty }}\delta
^{\tau -t}\left[ \left( 1-\alpha _{0}\right) +\left( \alpha _{0}-R_{\tau
}-r_{\tau }\right) \right] \left[ v(q_{\tau }^P)-x_{\tau }^P\right]\nonumber \\
&+\left( 1-\delta \right) \sum_{\tau = t}^{\infty}\delta^{\tau - t} \left[\sum_{T = 0}^{\tau} r_{\tau} \left[ v(q_{\tau}^{S}\left( T \right) )-x_{\tau}^{S}\left( T \right) \right]\right],
\end{align}
where the first term sums over the principal's profit from the agents who take the pooling contracts (which include all high types and some low types who haven't yet revealed), and the second term captures her profit from the agents who have revealed at different times in the past.

\subsubsection*{Constraints and Implementable Allocations}
In our environment with limited commitment, voluntary revelation, and voluntary exit, an allocation can be implemented only if: \textit{i)} the principal never wants to renege on her promises; \textit{ii)} information revelation is incentive-compatible for the low types, while the high types never want to make a false announcement; \textit{iii)} the agents' continuation payoffs never fall below their outside options, on and off the equilibrium path. These considerations lead to the following sets of constraints.

    \paragraph{Non-reneging constraints:}
    \begin{equation}\label{eq:nonreneging}
        \Pi_t \geq R_{t}\pi _{L}\left( q_{L}^{\ast }\right) + \left( 1-R_{t}\right)\pi _{H}\left( q_{H}^{\ast }\right)\text{, for all }t. \tag{NR}
    \end{equation}
    
    While the principal has to honor a one-period contract once it is accepted,\footnote{That is, if an agent is given a contract $(x_t,q_t)$ in period $t$, then the payment $x_t$ must be made as long as $q_t$ is produced.} she may renege on her past promises at the beginning of a period by presenting a contract that is different from what she has promised to offer.  We assume that the principal's reneging is observable to all agents, which makes it possible to punish such behavior. To understand the RHS, if the principal were to renege on her previous promises in period $t$, the worst punishment the agents can impose on the principal would be to cut off all future information. In this case, the best the principal can do after reneging is to extract the maximum surplus from the low-type agents who have revealed and, because of Assumption \ref{ass:1}, treat all unrevealed agents as if they were the high type.\footnote{
    The principal can always guarantee herself a payoff arbitrarily close to the RHS of \eqref{eq:nonreneging} by offering $\left(q_L^*,C_L(q_L^*)+\varepsilon\right)$ to each revealed agent and $\left(q_H^*,C_H(q_H^*)+\varepsilon\right)$ to each unrevealed agent in every future period. In Appendix \ref{app:B}, we explicitly construct a continuation equilibrium that generates this payoff.
    }
    
    We note that the RHS of \eqref{eq:nonreneging} is driven by two critical assumptions: \textit{i)} that the principal's reneging on any agent is publicly observable to all agents; and \textit{ii)} that low-type agents who haven't announced their type will react by stopping revealing information altogether. Technically, these assumptions lead to the worst punishment that agents as a group can impose on the principal. Economically, assuming the principal's reneging to be publicly observable captures the idea that workers often talk/complain to each other especially when they are mistreated by their manager. 
    Meanwhile, the agent's reaction to reneging could come from the following reasoning: once a worker sees his colleagues being mistreated by the manager, he will lose trust in the manager and believe that he would be treated the same way if he were to reveal his efficiency.\footnote{This reaction by the agent is indeed sequentially rational if the principal, upon reneging, would subsequently fully exploit anyone who later reveals to be the low-cost type.}
    
    \paragraph{Incentive compatibility constraints:\\}

    For the low type,
    \begin{align}\label{eq:IC-L}
        U^S_{L,t}(t) \geq \sup_{\mathbb{T}\geq t}(1-\delta) \sum_{\tau =t}^{\mathbb{T}}\delta ^{\tau -t}\left[ x_{\tau }^{P} -C_L\left(q_{\tau }^{P} \right)\right] + \delta^{\mathbb{T}} U^S_{L,\mathbb{T}}(\mathbb{T}), \text{ for all $t$ such that } r_t > 0,\nonumber\\
        \text{with equality if $0<r_t<\alpha-R_t$.}
        \tag{IC-L}
    \end{align}
\eqref{eq:IC-L} requires that, whenever a positive mass of agents makes the announcement in period $t$, any low-type agent should prefer to take the separating contracts beginning that period, rather than revealing at some optimal time later (perhaps never). Moreover, if some but not all low-type agents reveal in period $t$, then any low-type agent must be indifferent.

    For the high type,
    \begin{equation}\label{eq:IC-H}
        U^P_{H,t} \geq U^S_{H,t} : = \sup_{\mathbb{T}\geq t}\left( 1-\delta \right)
\sum_{\tau =t}^{\mathbb{T}}\delta ^{\tau -t}\left[ x_{\tau }^{S}\left(
t\right) -C_H\left(q_{\tau }^{S}\left( t\right) \right)
\right], \text{ for all }t. \tag{IC-H}
    \end{equation}
\eqref{eq:IC-H} insists that a high-type agent always
prefers to take the pooling contract, instead of falsely announcing his type in period 
$t$ followed by an optimal quitting decision. This constraint hedges against the \textit{``take-the-money-and-run"} strategy, in which a high type may reap an attractive initial payment in a separating contract and leave the relationship afterwards.

\paragraph{Participation constraints:}
\begin{align}
    U_{H,t}^P &\geq 0,\text{ for all }t. \tag{IR-H} \label{eq:IR-H} \\
U_{L,\ t+s}^{S}\left( t\right) &\geq 0,\text{ for all }t\text{ such that }r_{t} >0 \text{ and for all }s \geq 0. \tag{IR-L} \label{eq:IR-L}
\end{align}

\begin{definition}
    An allocation, $\left\{\left(q_t^P,x_t^P\right), \left(q^S_\tau(t),x^S_\tau(t)\right), r_t\right\}_{t\geq 0, \tau\geq t}$, is \textbf{implementable} by a sequential revelation mechanism if it satisfies \eqref{eq:nonreneging}, \eqref{eq:IC-H}, \eqref{eq:IC-L}, \eqref{eq:IR-H}, and \eqref{eq:IR-L}. 
\end{definition}

\subsubsection*{The Principal's Problem}
The principal aims to maximize her profit within the class of sequential revelation mechanisms. That is, she looks for an implementable allocation that maximizes her total discounted profit in period $0$:
\begin{align}\label{eq:principal}
\begin{aligned}
        \sup_{\left\{\left(q_t^P,x_t^P\right), \left(q^S_\tau(t),x^S_\tau(t)\right), r_t\right\}_{t\geq 0, \tau\geq t}} &\Pi_0 \\
        \text{s.t. }\eqref{eq:nonreneging},\  \eqref{eq:IC-H},\  \eqref{eq:IC-L},\  \eqref{eq:IR-H}, &\text{ and } \eqref{eq:IR-L}.
\end{aligned}
\end{align}

\begin{theorem}\label{thm:existence}
There is a solution to \eqref{eq:principal}. That is, an optimal mechanism exists.
\end{theorem}

\section{Properties of Optimal Mechanisms}\label{sec:structure}
We first define some terminology regarding the information revealed in a sequential revelation mechanism.
\begin{definition}
    A sequential revelation mechanism is \textbf{revealing} if $r_t>0$ for some $t$.  A sequential revelation mechanism is \textbf{non-revealing}  if $r_t=0$ for all $t$.
\end{definition}

\begin{definition}
    For a revealing mechanism, we say that information revelation is \textbf{gradual} if $r_t>0$ for all $t$, and that information revelation is \textbf{eventual} if $
    \sum_{t=0}^\infty r_t =\alpha_0$.
\end{definition}

Given Assumption \ref{ass:1}, if the principal does not seek to induce any information revelation, the best she can do is to offer the pooling contract $\left(q_H^*,C_H(q_H^*)\right)$ in every period to all agents, obtaining profit $\pi_H(q_H^*)$. In this context, we are specifically interested in the following questions: (When) Does an optimal mechanism involve information revelation? And if so, what are the qualitative features of the contracts offered to the agents? In this section, we characterize the structure of optimal contracts, \textit{assuming that information revelation can be fostered;} in the next section, we provide a necessary and sufficient condition for optimal mechanisms to be revealing.

In what follows, we assume that the maximal profit is strictly greater than $\pi_H(q^*_H)$ (thus non-revealing mechanisms are suboptimal),\footnote{It can be shown that when information revelation is possible, this condition holds generically.} and that the discount factor $\delta$ is greater than some threshold $\bar{\delta}$ which we shall define later.

\begin{lemma}[Frontloading of Payment/Rent to L]\label{lem:frontloading}
In any optimal mechanism, the sequence of separating contracts $\left\{\left(q_\tau^S(t),x_\tau^S(t)\right)\right\}_{\tau\geq t}$ for agents who reveal in period $t$ satisfies
\begin{equation}
        x_\tau^S(t) =\begin{cases}
        C_H(q_\tau^S(t)), &\text{ if } t\leq \tau < t+s\\
        C_L(q_\tau^S(t)) + \beta \Delta C (q_\tau^S(t)), &\text{ if } \tau = t+s\\
        C_L (q_\tau^S(t)), &\text{ if } \tau > t+s\\
        \end{cases}
\end{equation}
for some $s\geq 1$ and $\beta \in (0,1]$.
\end{lemma}
Lemma \ref{lem:frontloading} indicates that the separating contracts entail two phases: a \textit{reward phase} and an \textit{exploitation phase}. In the reward phase, the principal pays a positive rent to a low-type agent for at least two periods once he reveals his type. After the initial $s+1$ rent-paying periods, the exploitation phase starts and the rent disappears forever. This implies that an agent's continuation payoff after revealing declines over time and reaches zero in a finite number of periods.  We call each period in the reward phase a \textbf{debt period}, and we say that the principal is \textbf{in debt} with an agent in period $t$ if the agent's continuation payoff is positive.

To understand Lemma \ref{lem:frontloading}, note that the principal can always keep an agent's present value while adjusting the timing of payment to compensate him as much as possible in the initial periods following his type announcement. 
By fulfilling more today and promising less in the future, such frontloading of payment relaxes the principal's non-reneging constraint \eqref{eq:nonreneging} and makes room for a strict profit improvement.

However, there is a limit to how much rent can be paid to the low type in each period: If the principal pays too much in the initial periods of the separating contracts, the high type would be tempted to make a false announcement, work only in the first period to reap the high payment, and then exit the relationship; that is, \eqref{eq:IC-H} would be violated. In any optimal mechanism, the payoff of the high-type is always kept to be zero.  This implies that the principal can pay the low type at most $C_H(q_\tau^S(t))$ in any period without inducing the high type to take the money and run. By doing so, the principal can give at most $\Delta C(q_\tau^S(t))$ of rent to a low type in a period. When the discount factor is sufficiently large, the rent from each individual period matters little, so the principal has to use multiple debt periods to deliver the low-type's total information rent (i.e., $s + 1\geq 2$).\footnote{In any optimal mechanism, the low type's total information rent $U_{L,t}^S$ is equal to his payoff if he were never to reveal. This is positive and bounded away from zero, because by taking the pooling contract a low type gets paid $C_H(q_t^P)$ in each period while only incurring a cost of $C_L(q_t^P)$.}

\begin{lemma}[Gradual and Eventual Information Revelation]\label{lem:info} 
In any optimal mechanism, information revelation is gradual and eventual.
\end{lemma}
When an optimal mechanism is revealing, Lemma \ref{lem:info} indicates that in every period there will be a positive measure of low-type agents revealing themselves, and that every low-type agent will reveal his type in some period.

To see why information revelation must be gradual, note first that there cannot be a final period of revelation after which no further revelation occurs. Otherwise, the principal would not want to pay any information rent to any low types after that final revealing period. But this is inconsistent with the fact that the principal has to pay rent for at least two periods to the agents who reveal in that final period. Moreover, there cannot be any gaps between the periods with information revelation. Indeed, if $r_t$ and $r_{t+2}$ are positive but $r_{t+1}=0$, then the principal's non-reneging constraint \eqref{eq:nonreneging}  must be slack at $t+2$. This is because if the principal were to renege, he would rather renege at the beginning of $t+1$ than at $t+2$, as reneging one period earlier would save the rent she has to pay at $t+1$ without losing extra information (given that $r_{t+1} = 0$); in other words, the principal's incentive to renege is strictly less at $t+2$ than at $t+1$. But then, one can shift a small measure of agents who were supposed to reveal at $t+2$ to reveal at $t+1$, without violating \eqref{eq:nonreneging} at $t+2$ (or any other constraint).\footnote{In general, by moving agents to reveal earlier, it will hurt the non-reneging constraints in later periods because there is less future information to lose. The slackness of \eqref{eq:nonreneging} at $t+2$ makes it possible to induce earlier revelation without violating \eqref{eq:nonreneging}.} This variation successfully induces some low-types agents to reveal earlier and strictly benefits the principal.

Information revelation must also be eventual because if $\sum_{t=0}^\infty r_t <\alpha_0$, one can create another implementable allocation with the same set of contracts while scaling up the measure of the revealing agents in every period. Doing so will not violate any constraint but will induce more information to be revealed over time, thereby increasing the principal's profit.
\begin{lemma}[Downward Distortion of H's Quantity]\label{lem:poolingquantity}
In any optimal mechanism, the sequence of pooling contracts $\left\{\left(q_t^P,x_t^P\right)\right\}_{t\geq 0}$ satisfies $q_t^P<q_H^*$ for all $t$ and $\lim_{t\to\infty}q_t^P = q_H^*$. Moreover, $x_t^P = C_H(q_t^P)$.
\end{lemma}

Lemma \ref{lem:poolingquantity} states that in any optimal mechanism, the high-type agents are asked to produce less than their surplus-maximizing quantity in each period and the quantity produced converges to the surplus-maximizing level in the long run. This is an analogue of the standard distortion in screening models: the low-type agents in our environment are given information rent to reveal their type, and the total rent each of them gets depend on what he would have got if he were never to reveal. In an optimal mechanism, the pooling contracts always pay $C_H(q_t^P)$ to give zero payoffs to the high types. This means that by taking the pooling contract in period $t$, a low-type agent can obtain a positive payoff of $\Delta C(q_t^P)$ because he only incurs a cost of $C_L(q_t^P)$. Therefore, lowering $q_t^P$ from $q_H^*$ is profitable, because it strictly reduces the rent having to be paid to the low types while only having a second-order effect on the principal's profit from the high types. The convergence of $q_t^P$ to $q_H^*$ over time follows from Lemma \ref{lem:info}, as all low-type agents opt out of the pooling contracts eventually. 

\begin{lemma}[Upward Distortion and Frontloading of L's Quantity]\label{lem:separatingquantity}
In any optimal mechanism, the sequence of separating contracts $\left\{\left(q_\tau^S(t),x_\tau^S(t)\right)\right\}_{\tau\geq t}$ for agents who reveal in period $t$ satisfies: If $\{t, t+1, ..., t+s\}$ is the set of debt periods, then $q_t^S(t) > q_{t+1}^S (t) > ... > q_{t+s}^S (t) \geq q_L^{*}$ and $q_{\tau}^S(t) = q_L^*$ for all $\tau > t+s$.
\end{lemma}

Lemma \ref{lem:separatingquantity} implies that the principal never asks a low-type agent to produce below the surplus-maximizing level and will distort the quantities upward in all (non-final) debt periods. 
Moreover, the quantities required in the separating contracts are decreasing over time throughout the debt periods.

To understand the direction of distortion, note that if $q_\tau^S(t)< q_L^*$, the principal can increase $q_\tau^S(t)$ and $x_\tau^S(t)$ at the same time to keep the low-type agent's payoff unaffected while improving the surplus and her profit. 
Moreover, further increasing $q_\tau^S(t)$ above $q_L^*$ in a debt period can help the principal frontload information rent and thus relax her non-reneging constraints in future periods.

The quantities in the separating contracts must be decreasing across the debt periods, because otherwise if $q_\tau^S(t)\leq q_{\tau+1}^S(t)$ for a debt period $\tau$, the principal can then increase $\left(q_\tau^S(t), x_\tau^S(t)\right)$ by $(\varepsilon, \varepsilon C_H')$ and decrease $\left(q_{\tau+1}^S(t), x_{\tau+1}^S(t)\right)$ by $(\delta^{-1}\varepsilon, \delta^{-1}\varepsilon C_H')$, leaving each type's payoff unchanged. This perturbation is profitable, as it leads to a higher average surplus from the low types (due to the concavity of the surplus function) and also allows the principal to frontload some information rent from $\tau+1$ to $\tau$.\footnote{Indeed, this perturbation increases the low type's payoff by $\varepsilon\Delta C'$ at $\tau$ and decreases it by $\delta^{-1}\varepsilon\Delta C'$ at $\tau+1$.}  

Combining these results, the following theorem characterizes the qualitative structure of the contracts offered in an optimal mechanism.
\begin{theorem}\label{thm:structure}
    Suppose that non-revealing mechanisms are suboptimal. 
    There exists a $\bar{\delta}$ such that if $\delta>\bar{\delta}$, then in every optimal mechanism,
    \begin{enumerate}
        \item Information revelation is gradual and eventual;
        \item After a low-type agent reveals his type, both the amount of work demanded and the payment he receives start high and decrease over time;
        \item A low-type agent is given positive rents for a finite number of periods following his revelation, after which the principal extracts full surplus from him;
        \item The amount of work demanded in the pooling contracts starts below the
surplus-maximizing level and converges to it over time.
    \end{enumerate}
\end{theorem}

\begin{remark}
    Fixing the parameters, there may be multiple optimal mechanisms, but our qualitative characterization applies to every one of them. Moreover, a revealing mechanism and a non-revealing mechanism cannot be optimal at the same time.
\end{remark}
\begin{remark}
    Our analysis allows each sequence of separating contracts to be tailored to the revealing date. Restricting these sequences to be independent of such date does not change our results: the characterization in Theorem \ref{thm:structure} holds verbatim and the condition in Theorem \ref{thm:nsc} remains the same.
\end{remark}

{
Theorem \ref{thm:structure} summarizes the testable implications of our model for the work and wage dynamics in equilibrium. In particular, it makes sharp predictions about the life cycle of an agent's contract, depending on his type. As illustrated in Figure \ref{fig:lifecycle}, for an inefficient agent, his wage and work amount gradually increase over time, but both of them remain at a relatively low level throughout the game. On the other hand, an efficient agent's contract exhibits an inverse U-shaped pattern. His wage and work amount may start low, but once the principal learns about his competence, he will receive a major promotion that comes with a significant increase in both wage and responsibilities. However, due to the principal's limited commitment, such high-powered incentives providing information rent will not last forever: payments and quantities will gradually fall to the surplus-maximizing levels, at which point the principal starts to fully exploit the agent.

\begin{figure}[ht]
	\centering
	\begin{subfigure}[t]{.49\linewidth}
 \includegraphics[width=1\linewidth]{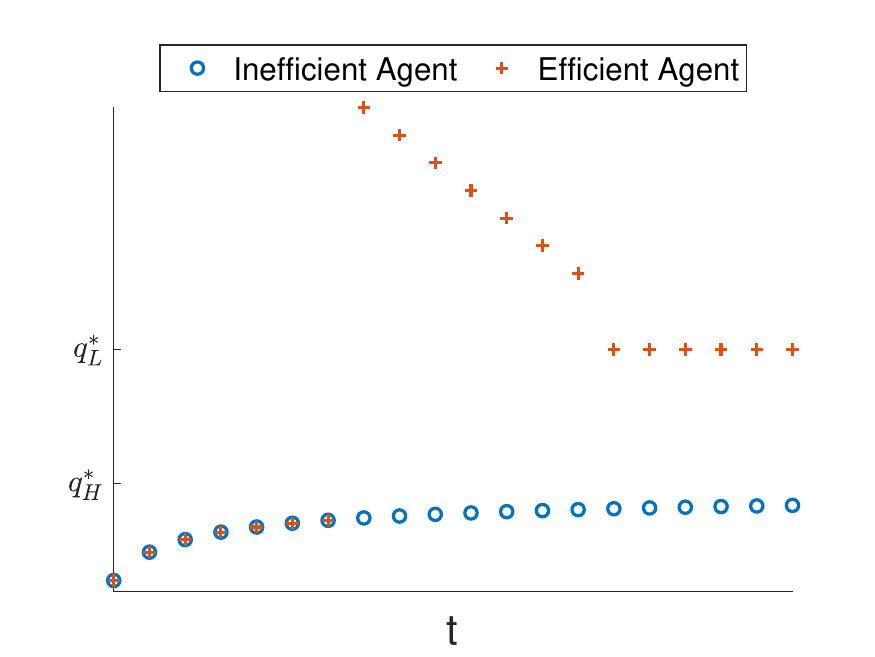}
		\caption{Type-Dependent Production Dynamics}
	\end{subfigure}
	\begin{subfigure}[t]{.49\linewidth}
 \includegraphics[width=1\linewidth]{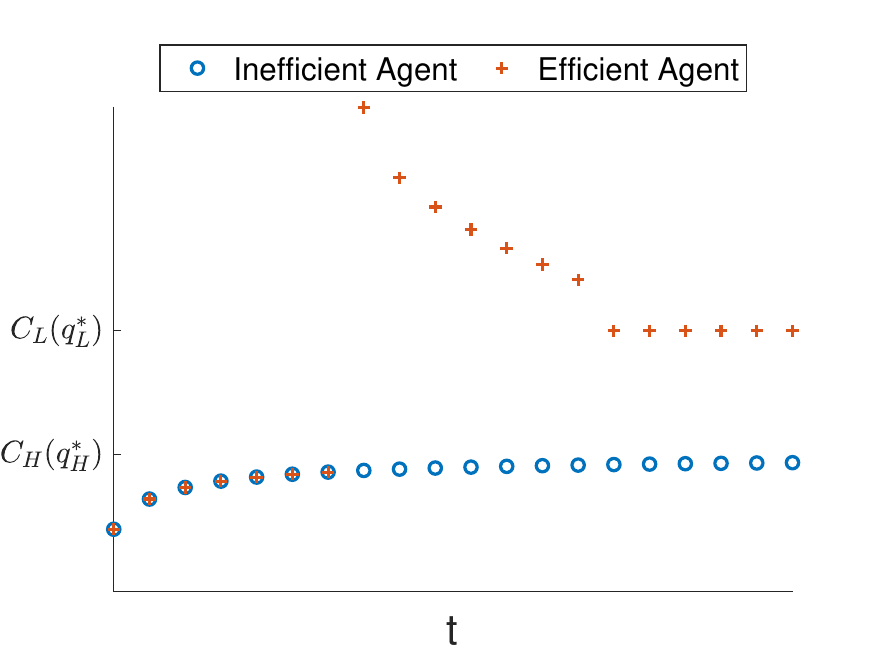}
		\caption{Type-Dependent Payment Dynamics}
	\end{subfigure}
	\caption{Life Cycle of an Agent's Contract. This figure is plotted under the following functional forms and parameter values: $v(q) = 2q^{1/2}$, $C_\theta (q) = \theta q$, $\theta_H=3$, $\theta_L = 2$; the separating contracts are for an efficient agent who reveals his type at $t=8$.}
	\label{fig:lifecycle}
\end{figure}

In industries where contracts are subject to renegotiation and skills are project-specific, we do observe such inverse U-shaped patterns for capable workers. For example, IT consultants and contractors, after demonstrating their capabilities of implementing certain projects, are typically hired with remunerative contracts for initial infrastructure setups. These contracts are usually renegotiated with reduced rates once the project enters a maintenance-focused phase.\footnote{In 2022, Gartner highlighted and advised that as companies moved toward routine maintenance phases post-initial IT infrastructure setup, many renegotiated contracts to reflect lower rates (see Kasey Panetta, ``IT Budgets Are Growing. Here's Where the Money's Going," \textit{Gartner}, October 21, 2021; Lori Perri, ``Trying to Cut Software and Cloud Costs? Here's How," \textit{Gartner}, December 16, 2022).} 
In the entertainment industry, actors in multi-season TV shows/movies are offered lucrative contracts once they prove their fit for certain characters. However, pay cuts are sometimes renegotiated in later seasons, despite a show's continued success and an actor's consistently good performance.\footnote{For example, the main cast of \textit{``The Simpsons"} agreed to a significant salary reduction in their 2011 contract renegotiation with Fox in order to keep the show on the air (see Alex Ben Block and Kim Masters, ``Simpsons Cast Blinks in Salary Showdown With Fox," \textit{The Hollywood Reporter}, October 7, 2011). In another instance, Stephen Amell, a leading actor in \textit{``Arrow,"}  noted in multiple interviews he was prepared to continue working at a potentially ``reduced pay structure" for future seasons. The last season of the show moved to shorter 10-episode runs instead of the usual 22-episode seasons.} In research-intensive universities, professors are promoted with tenure after demonstrating their strong research ability. Over time, tenured faculty face less pressure on research but are often assigned research-unrelated administrative work which could be argued as effectively reducing their real pay.

One should not take these predictions only at face value, as our model abstracts away from other realistic features, such as increasing productivity over time, which naturally lead to rising work and wages. Instead, our analysis serves to identify two important economic forces on screening contracts in environments with limited commitment. To induce competent workers to reveal, they have to be rewarded with sufficient information rent, leading to the initial raises in pay and responsibilities after information disclosure. But since the principal is unable to commit to not using information about a worker's type forever, the rent will disappear after some periods, resulting in a negative force on quantities and wages. If productivity also increases over time, we may not necessarily observe actual declines in compensation, but rather the flattening of it following a major raise, as seen with corporate executives and law/consulting firm partners.
}

\section{When Is Information Revelation Feasible?}
The analysis in the previous section presumes that the optimal mechanism is revealing. We now turn to deriving conditions under which this presumption is true. Our result will uncover a new tradeoff between fostering future information revelation and exploiting past disclosures.
\subsection{An Intuitive Derivation}\label{sec:intuitivederivation}
Note that the key to overcoming the ratchet effect in our setting is to make sure that the principal does not want to renege on paying the information rent promised to the low-type agents. This is the premise for any low types to be willing to reveal their private information. 

We begin with some heuristic calculations to illustrate the main tradeoffs faced by the principal.  Consider a stationary auxiliary environment where: \textit{i)} a constant measure  $r$ of low-cost agents reveal in each period; \textit{ii)} the principal does not distort the quantity demanded from the high-cost agents; and \textit{iii)} for $T$ debt periods after a low-type agent reveals, the principal demands a fixed quantity $\tilde{q}_L\in [q_L^*,\bar{q}]$ while paying him $C_H(\tilde{q}_L)$. These features mirror the key properties of an optimal revealing mechanism (i.e., Lemmas \ref{lem:frontloading} through \ref{lem:separatingquantity}), at least asymptotically, while the stationarity of this auxiliary environment simplifies the examination of the principal's reneging incentives.

Specifically, in this stationary auxiliary environment, reneging is never optimal for the principal if and only if waiting one more period to renege is better than reneging right away. From the principal's perspective, if he waits one more period and then reneges, an additional measure $r$ of low-type agents will reveal in the current period, enabling the principal to make an extra profit of 
\begin{align*}
    r\delta\left[\pi_L(q_L^*) - \pi_H(q_H^*)\right]
\end{align*}
by exploiting them starting the next period. 

On the other hand, waiting one more period to renege requires the principal to pay the promised rent to those agents with whom she is still in debt. Because in period $t$ the principal is in debt with agents who revealed at $t-1$, $t-2$, ..., $t-T$ and each of these agents gets paid $C_H(\tilde{q}_L)$ while producing $\tilde{q}_L$, the total loss of profit from these agents in the \textit{current} period if the principal honors her promises is\footnote{The loss of profit per agent, $\pi_L(q_L^*)-\pi_H(\tilde{q}_L)$, is equal to the information rent paid to a low-type agent, $\Delta C(\tilde{q}_L)$, plus the inefficiency loss due to quantity distortions, $\pi_L(q_L^*)-\pi_L(\tilde{q}_L)$.}
$$(1-\delta)Tr\left[\pi_L(q_L^*)-\pi_H(\tilde{q}_L)\right].$$

Note further that the number of debt periods $T$ is endogenous as it is determined by the number of periods needed to deliver a low-type's total information rent. Indeed, $T$ has to satisfy
\begin{align*}
    (1-\delta)\sum_{\tau=0}^{T-1}\delta^\tau\Delta C(\tilde{q}_L) = \Delta C(q_H^*)
\end{align*}
where the RHS is the total information rent needed to induce a low-type to reveal (equaling his payoff if he were to take the pooling contract forever) and the LHS is the discounted sum of his rents from the $T$ debt periods after revealing. When $\delta$ approaches one, this condition implies that\footnote{To see this, note that $ (1-\delta)\sum_{\tau=0}^{T-1}\delta^\tau\Delta C(\tilde{q}_L) = \Delta C(q_H^*)$ implies that $1-\delta^T = \frac{\Delta C({q}_H^*)}{\Delta C(\tilde{q}_L)}$, and thus $(-\ln \delta) T =\ln \left(\frac{\Delta C(\tilde{q}_L)}{\Delta C(\tilde{q}_L)-\Delta C({q}_H^*)}\right)$. Moreover, when $\delta$ is close to $1$, $-\ln\delta$ is approximately $(1-\delta)$.}
\begin{align*}
    (1-\delta)T \to \ln \left(\frac{\Delta C(\tilde{q}_L)}{\Delta C(\tilde{q}_L)-\Delta C({q}_H^*)}\right).
\end{align*}
Therefore, from the principal's perspective, as $\delta$ approaches one, waiting one more period to renege is strictly better than reneging right away if and only if
\begin{align}\label{eq:nscondition}
    \pi_L(q_L^*) - \pi_H(q_H^*) > \ln \left(\frac{\Delta C(\tilde{q}_L)}{\Delta C(\tilde{q}_L)-\Delta C({q}_H^*)}\right) \left[\pi_L(q_L^*)-\pi_H(\tilde{q}_L)\right];
\end{align}
that is, when the benefit of waiting outweighs the loss.

\subsection{The Formal Result}
While the above analysis focuses on an auxiliary environment, it turns out that the same condition determines whether information revelation is feasible in our original non-stationary setting. We have the following theorem.
\begin{theorem}\label{thm:nsc}
\mbox{ }
\begin{enumerate}
    \item If for all $\tilde{q}_L \in [q_L^*,\bar{q}]$,
\begin{align}\label{eq:nsconditionnonrevealing}
        \pi_L(q_L^*) - \pi_H(q_H^*) < \ln \left(\frac{\Delta C(\tilde{q}_L)}{\Delta C(\tilde{q}_L)-\Delta C({q}_H^*)}\right) \left[\pi_L(q_L^*)-\pi_H(\tilde{q}_L)\right],
\end{align}
then the (unique) optimal mechanism is non-revealing whenever $\delta$ is sufficiently close to $1$.

\item If for some $\tilde{q}_L \in [q_L^*,\bar{q}]$,
\begin{align}\label{eq:nsconditionrevealing}
        \pi_L(q_L^*) - \pi_H(q_H^*) > \ln \left(\frac{\Delta C(\tilde{q}_L)}{\Delta C(\tilde{q}_L)-\Delta C({q}_H^*)}\right) \left[\pi_L(q_L^*)-\pi_H(\tilde{q}_L)\right],
\end{align}
then every optimal mechanism is revealing whenever $\delta$ is sufficiently close to $1$.
\end{enumerate}
\end{theorem}

Below we interpret this theorem economically and use it to analyze the effect of skill heterogeneity on information revelation. For readers interested in the technical details, we sketch the proof of this theorem in Section \ref{sec:proofsketch}.

\subsection{Economic Implications}
\subsubsection*{Cost-Benefit Comparison}
Economically, the LHS of condition \eqref{eq:nsconditionrevealing} in Theorem \ref{thm:nsc} is equal to the maximum gains in surplus after the principal obtains information about an agent's type, while the RHS measures the costs to the principal for inducing separation of types. Specifically, to incentivize information revelation, the principal needs to distort the low-type's quantity (Lemma \ref{lem:separatingquantity}) and pay information rent for a finite number of debt periods (Lemma \ref{lem:frontloading}). Based on the previous calculations, $\pi_L(q_L^*)-\pi_H(\tilde{q}_L)$ measures the principal's profit loss per agent per debt period, and the multiplier, $\ln \left(\frac{\Delta C(\tilde{q}_L)}{\Delta C(\tilde{q}_L)-\Delta C({q}_H^*)}\right)$, accounts for the \textit{discounted} number of debt periods. Theorem \ref{thm:nsc} says that an optimal mechanism is revealing if and only if the social value of information exceeds the costs needed to incentivize a low-type agent to voluntarily reveal his private information.

A useful observation from Theorem \ref{thm:nsc} is that, for an optimal mechanism to be revealing, the discounted number of debt periods, $\ln \left(\frac{\Delta C(\tilde{q}_L)}{\Delta C(\tilde{q}_L)-\Delta C({q}_H^*)}\right)$, must be less than $1$. Otherwise, it is impossible to use a smaller surplus gain (per period) to cover a larger separation cost (per period).
\begin{observation}\label{obs:multiplier}
If an optimal mechanism is revealing, then $\ln \left(\frac{\Delta C(\bar{q})}{\Delta C(\bar{q})-\Delta C({q}_H^*)}\right)<1$.
\end{observation}
This condition requires $\Delta C(q_H^*)$, the low type's total information rent (i.e., his payoff if he were to take the pooling contract forever), to be not too large, or $\Delta C(\bar{q})$, the maximal information rent deliverable in a given period, to be sufficiently large.

Graphically, Figure \ref{fig:nsc} depicts the conditions in Theorem \ref{thm:nsc}. The blue curves plot the RHS of \eqref{eq:nsconditionrevealing} as a function of $\tilde{q}_L\in [q_L^*,\bar{q}]$, and the yellow lines indicate the values of the LHS which are constant w.r.t. $\tilde{q}_L$. Under the parameters of panels (a) and (b) (where $\theta_L$ and $\theta_H$ are sufficiently distant from each other), there exists some $\tilde{q}_L$ at which the yellow line is above the blue curve, thus every optimal mechanism is revealing. In contrast, under the parameters of panel (c) (where $\theta_L$ and $\theta_H$ are close), the yellow line is always below the blue curve, so the optimal mechanism is non-revealing.
\begin{figure}[ht]
	\centering
	\begin{subfigure}[t]{.49\linewidth}
		\includegraphics[width=1\linewidth]{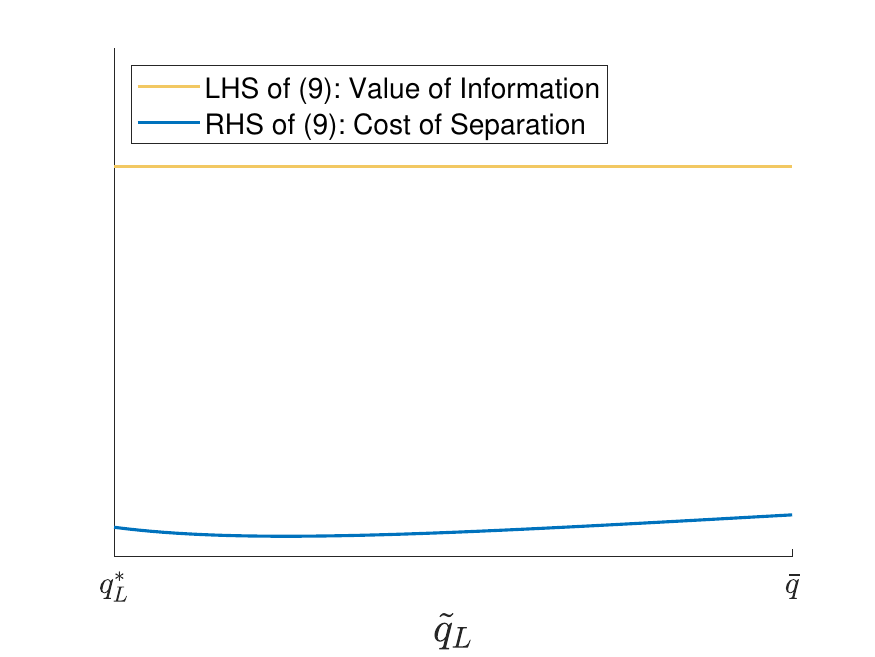}
		\caption{$\theta_L= 1.5,\ \theta_H = 3$ (Revealing)}
	\end{subfigure}
	\begin{subfigure}[t]{.49\linewidth}
		\includegraphics[width=1\linewidth]{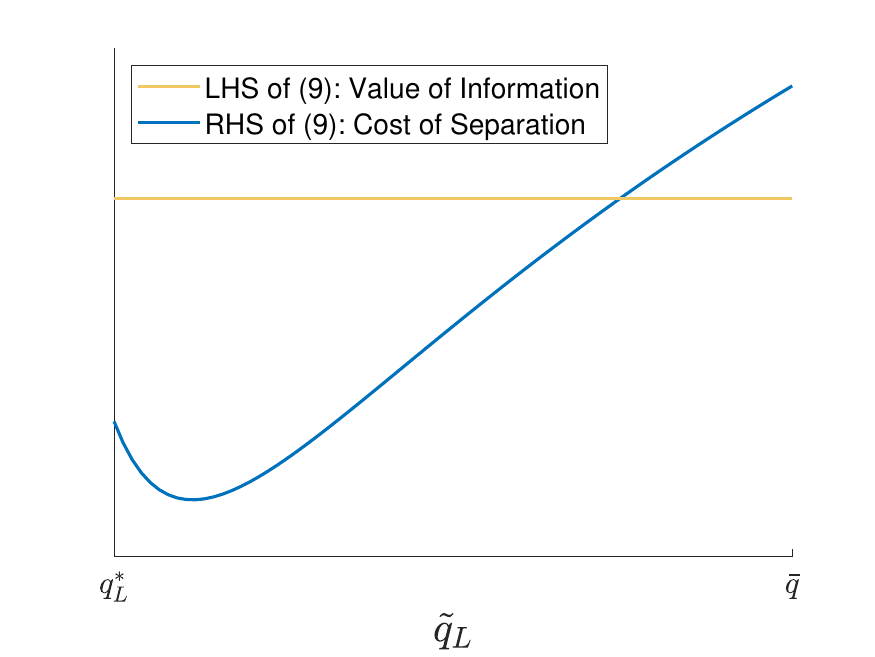}
		\caption{$\theta_L= 2,\ \theta_H = 3$ (Revealing)}
	\end{subfigure}
    \begin{subfigure}[t]{.49\linewidth}
		\includegraphics[width=1\linewidth]{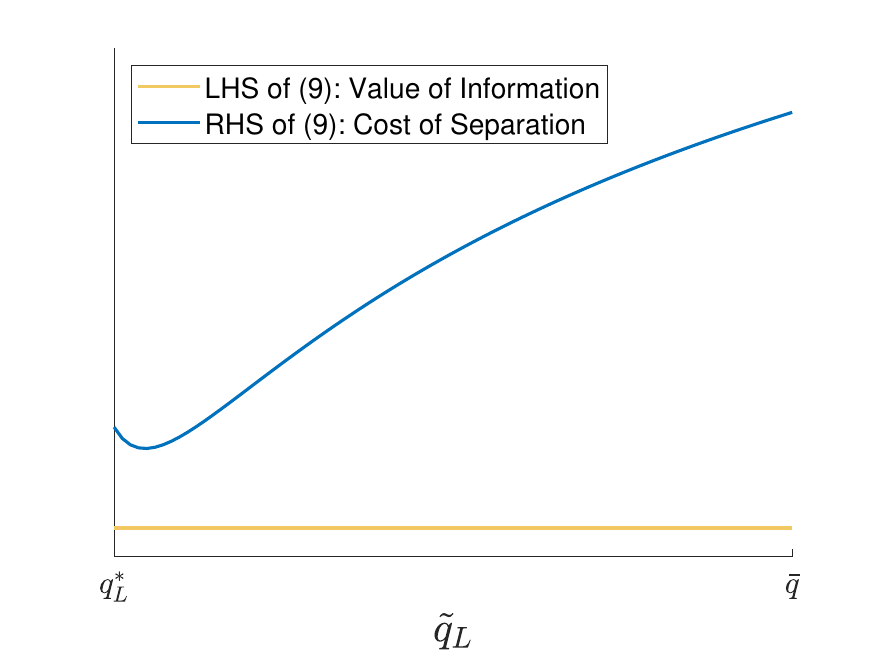}
		\caption{$\theta_L= 2.5,\ \theta_H = 3$ (Non-Revealing)}
	\end{subfigure}
	\caption{Necessary and Sufficient Condition for Information Revelation. This figure is plotted under the following functional forms and parameter values: $v(q) = 2q^{1/2}$, $C_\theta (q) = \theta q$, $\bar{q} = 1$, $\theta_H=3$, $\theta_L = 1.5$ for (a), $\theta_L=2$ for (b), and $\theta_L=2.5$ for (c).}
	\label{fig:nsc}
\end{figure}

\subsubsection*{Skill Heterogeneity Facilitates Information Revelation}
How does skill heterogeneity among agents affect information revelation? On the one hand, greater skill heterogeneity increases the value of information about an agent's type. On the other hand, as the cost difference widens, the low type can get a higher payoff by pretending the high type, so the principal has to pay more information rent to induce separation of types. Figure \ref{fig:nsc} suggests that, overall, information revelation becomes easier when the agents' skills are more heterogeneous, that is, when $\theta_H$ and $\theta_L$ are further apart. We now formalize this prediction. 

Let us first give some cardinal meaning to the agent's type. To this end, we assume
\begin{align*}
    C_\theta(q) = \theta q.
\end{align*}
We have the following result.
\begin{proposition}\label{prop:typedistance}
Fix all parameters except $\theta_L$ and $\delta$.\footnote{This result also requires $\bar{q}$ to be large enough (but fixed). See Appendix \ref{app:thm3} for details.} There exists $\bar{\theta}<\theta_H$ such that: 
\begin{itemize}
    \item if $\theta_L >\bar{\theta}$, then the (unique) optimal mechanism is non-revealing whenever $\delta$ is sufficiently close to $1$;
    \item if $\theta_L <\bar{\theta}$, then every optimal mechanism is revealing whenever $\delta$ is sufficiently close to $1$.
\end{itemize} 
\end{proposition}

Proposition \ref{prop:typedistance} indicates that information revelation is more likely when the difference between the two types' marginal costs gets larger, and that separation becomes impossible if the difference in marginal costs is too small. Indeed, when the gap in marginal costs vanishes, the value of information is negligible and is thus unable to cover any cost for incentivizing separation. On the other hand, if $\theta_L$ decreases (i.e., when the difference in marginal costs rises), both the surplus gains from information revelation and the cost to induce separation will increase by $q_L^*$. However, by Observation \ref{obs:multiplier}, the multiplier on the RHS of \eqref{eq:nsconditionrevealing} is less than $1$ whenever the condition is satisfied, making the surplus gains increase at a higher rate. Therefore, if an optimal mechanism is revealing at some $\bar{\theta}$, it remains so for any $\theta_L<\bar{\theta}$.

\subsection{Sketch of Proof}\label{sec:proofsketch}
Below we outline the key steps for proving Theorem \ref{thm:nsc}, especially how they relate to the principal's tradeoffs in the auxiliary environment analyzed in Section \ref{sec:intuitivederivation}.  The formal proof is in the appendix.

To establish the sufficiency of condition \eqref{eq:nsconditionrevealing} for information revelation, note that the biggest jump from the auxiliary environment to our model is the constant measure $r$ of revealing agents in every period. This is apparently infeasible in our model with a fixed mass of (low-type) agents. However, one could set the measure of low-type agents revealing in period $t$ to be $r_0\rho^t$. With $\rho$ close to (but less than) $1$ and $r_0$ close to $0$, the sequence is almost constant while the sum can still equal $\alpha_0$.\footnote{Recall that in the auxiliary environment, condition \eqref{eq:nscondition} for information revelation is independent of $r$. This gives us the freedom to adjust $r_0$ without affecting the principal's incentives.} When condition \eqref{eq:nsconditionrevealing} holds, this construction allows us to obtain implementable allocations where the principal's incentives are almost identical to the auxiliary environment.

Let us now turn to necessity. Assume that an optimal mechanism is revealing. Then in every period, the principal should prefer to honor her promise for one more period than reneging right away. As the quantity in the pooling contract converges to $q_H^*$ in the long run (by Theorem \ref{thm:structure}), for each sufficiently large $t\geq \tilde{t}$, we have
\begin{equation*}
r_t\delta\left[\pi_L(q_L^*)-\pi_H(q_H^*)\right] \geq (1-\delta)\sum_{\tau = z(t)}^{t-1}r_\tau\left[\pi_L(q_L^*) - \pi_H (q_t^S(\tau))\right].
\end{equation*}
The LHS is the benefit of honoring promise for one more period and then reneging (relative to reneging right away), since information about an extra $r_t$ measure of efficient agents will be revealed. The RHS represents the cost, as the principal will forgo some profits from those agents revealing in $\{z(t),...,t-1\}$ with whom she is still in debt.

Adding up these constraints from $\tilde{t}$ onward, we obtain
\begin{align*}
\sum_{t=\tilde{t} }^{\infty}r_t\delta\left[\pi_L(q_L^*)-\pi_H(q_H^*)\right] &\geq (1-\delta)\sum_{t=\tilde{t} }^{\infty}\sum_{\tau = z(t)}^{t-1}r_\tau\left[\pi_L(q_L^*) - \pi_H (q_t^S(\tau))\right]\\
& \geq
(1-\delta)\sum_{t=\tilde{t} }^{\infty}r_t\sum_{\tau = t+1}^{T(t)}\left[\pi_L(q_L^*) - \pi_H (q_\tau^S(t))\right]\\
& \geq 
(1-\delta)\sum_{t=\tilde{t} }^{\infty}r_t\sum_{\tau = t+1}^{T(t)}\left[\pi_L(q_L^*) - \pi_H (\tilde{q}_{L,t})\right], 
\end{align*}
where the second line follows from changing the order of summation and dropping terms from periods before $\tilde{t}$, and the third line follows from the concavity of $\pi$ and Jensen's inequality, with $\tilde{q}_{L,t}$ defined as an ``average" of $q_\tau^S(t)$ across the debt periods that delivers the same total information rent.\footnote{The formal proof also uses the concavity of $\Delta C$ in this step to ensure the desired direction of the inequality.}

Rearranging the above inequality, we obtain
\begin{equation*}
\sum_{t=\tilde{t} }^{\infty}r_t\left\{\delta\left[\pi_L(q_L^*)-\pi_H(q_H^*)\right]-(1-\delta)(T(t)-t) \left[\pi_L(q_L^*) - \pi_H (\tilde{q}_{L,t})\right]\right\}\geq 0.
\end{equation*}
This means that for at least some $t$, we must have
\begin{equation*}
    \delta\left[\pi_L(q_L^*)-\pi_H(q_H^*)\right]\geq (1-\delta)(T(t)-t) \left[\pi_L(q_L^*) - \pi_H (\tilde{q}_{L,t})\right].
\end{equation*}
Now the payoff comparison becomes very close to that in the auxiliary environment. As $\delta\to 1$, calculations similar to Section \ref{sec:intuitivederivation} imply that condition \eqref{eq:nsconditionnonrevealing} must fail, establishing part (1) of Theorem \ref{thm:nsc}.

\section{Discussion}
\subsection{The Role of Assumption \ref{ass:1}}\label{sec:alpha}
We have assumed that $\alpha _{0} <  \ushort{\alpha} \equiv \frac{\pi _{H}(q_{H}^{\ast })}{\pi
_{L}(q_{L}^{\ast })}$, that is, the initial proportion of the efficient agents is not too large. This assumption ensures that the inefficient agents are sufficiently important to the principal, and in particular, at any history the worst possible continuation equilibrium is obtained through the best
\textit{pooling} contract for the principal. 

On the other hand, it is easy to see
that there exists $\bar{\alpha}\in \left( 0,1\right) $ such that if $\alpha_0 >\bar{\alpha}$, the optimal solution with full commitment is to make an offer that extracts the entire surplus from the efficient type, while the inefficient type simply leaves the relationship. Since this allocation does not entail time inconsistency, it is also optimal in our setting with limited commitment.

What about $\alpha_0 \in \left[ \ushort{\alpha},\ \bar{\alpha}\right] ?$ One can show that our characterizations extend to a small region beyond these known boundaries. In particular, there exist $\varepsilon _{1}>0$ and $\varepsilon _{2}>0$
such that: when $\alpha \in \left[ \ushort{\alpha},\ \ushort{\alpha}
+\varepsilon _{1}\right] $, the optimal mechanism is revealing whenever condition \eqref{eq:nsconditionrevealing} in Theorem \ref{thm:nsc} is satisfied; when $\alpha \in \left[ \bar{\alpha}-\varepsilon _{2},\ \bar{\alpha}\right] $, it is optimal for the principal to offer only one contract that excludes
the inefficient type. Unfortunately, figuring out the exact values
of $\varepsilon _{1}$ and $\varepsilon _{2}$, as well as what happens for
priors in $\left(\ushort{\alpha} +\varepsilon _{1},\ \bar{\alpha}-\varepsilon _{2}\right)$, is a daunting
task that requires a different set of techniques. We leave it as a direction for future research.

\subsection{The Role of the Concavity of $\Delta C$}\label{sec:concavedeltac}
The assumption that $\Delta C$ is concave (i.e., $C_H''\leq C_L''$) is used to show that quantities in the separating contracts are decreasing during reward phase (Lemma \ref{lem:separatingquantity}) and to prove Theorem \ref{thm:nsc}. To see the role of this assumption, note that an efficient agent's total information rent is determined by his payoff from not revealing, which is captured by $\Delta C(q_H)$, i.e., receiving a payment equal to the high cost while incurring a low cost. Moreover, the maximum information rent deliverable per period is constrained by the inefficient agent's take-the-money-and-run incentive, so an efficient agent can be compensated by at most $C_H(q_L)$ and enjoys a rent of $\Delta C(q_L)$ per period. The assumption $\Delta C^{\prime \prime }\leq 0$ then ensures that the principal's problem is concave in quantities, which substantially simplifies the analysis.

When $\Delta C$ is not concave, the principal's problem may no longer be concave. Nevertheless, Theorem \ref{thm:nsc} still holds with a slight modification: the necessary and sufficient condition now involves an appropriate concavification of the principal's objective function. Furthermore, as $\delta $ goes to one, we can show that the quantities in the separating contracts still decreases over time in a probabilistic sense.

{
\subsection{The Cardinality of the Set of Agents}\label{sec:cardinality}
In our model, the principal interacts with a continuum of agents. This setting allows us to sharply demonstrate how the principal's desire to elicit future information could help overcome her temptation to excessively exploit past disclosures, thereby mitigating the ratchet effect, even in the absence of exogenously new information. That said, one may wonder whether our results extend to environments with ``fewer" agents.

First, we can show that whenever information revelation is feasible with a continuum of agents, there exists an environment with countably many agents in which information revelation is also feasible. This would naturally arise in an environment where a principal interacts with finitely many agents each period, with new agents arriving and replacing old agents over time. Consider a variation of our model in which the agents leave the relationship
at a constant rate $\lambda$ and are replaced by new ones over time. One can show that whenever condition \eqref{eq:nsconditionrevealing} from Theorem \ref{thm:nsc} holds, the optimal stationary mechanism is revealing, and it will exhibit the same qualitative structure as in Theorem \ref{thm:structure}. On the other hand, even if condition \eqref{eq:nsconditionrevealing} fails, information revelation is still possible when the replacement rate $\lambda$ is sufficiently high. Intuitively, replacements make information revelation easier because if an agent (with a known type) is to be replaced very soon by another one (with an unknown type), the principal's benefit from exploiting this agent is rather limited, while the cost of losing future information can still be substantial.

Nonetheless, if there is a \textit{fixed} pool of finitely many agents in our model, the unique equilibrium outcome when players are sufficiently patient will be pooling, i.e., sequential information revelation is impossible. This is because with finitely many agents, if there were information revelation, there would be some point in the future when all agents who plan to reveal have revealed their types. After the final revealing period, the principal would have no future information to learn/lose, so she would fully exploit everyone, but then no agent would be willing to reveal his private information in the final revealing period, leading to unraveling.\footnote{A similar argument goes through even when agents can reveal probabilistically. With finitely many agents, the expected number of agents revealing on the equilibrium path will get arbitrarily small as time goes on. At some point, the principal will find it optimal to start fully exploiting everyone who has revealed, instead of paying a bulk of information rent to squeeze additional information.}$^,$\footnote{The distinction between finite and infinite agents in this model is reminiscent of the fundamental difference between finitely and infinitely repeated games.
} Despite this technical prediction, we believe that the main economic forces identified in this paper are present as long as there is an indefinite need for new information, which could be driven by evolving market conditions, adoptions of novel technologies, and/or arrivals of fresh workers.  
}

\section{Concluding Remarks}
We study how simultaneously contracting with many agents can alleviate a principal's commitment problem and facilitate information revelation in the presence of the ratchet effect. The principal's reneging on previous promises, while tempting, will stop all future agents from revealing their private information. Her fear for losing future information then becomes a deterrent to the principal's exploitation of the agents who have revealed their types in the past. This in turn makes each agent more willing to reveal his type, thereby mitigating the ratchet effect. We fully characterize the structure of optimal contracts, and find that these contracts frontload information rent and create distortions at all types of agents. We also provide a necessary and sufficient condition that determines whether the agents' private information is never revealed or fully revealed in the long run; this condition uncovers a novel tradeoff between fostering future information revelation and exploiting past revelation.

Our analysis has focused exclusively on sequential revelation mechanisms. While we believe these mechanisms are suited to address our main research questions, without a formal revelation principle, one cannot rule out the possibility that some other mechanisms may further improve the principal's profit. We leave the examination of mechanisms outside of this class and the development of a more general revelation principle as interesting directions for future research.

\bibliographystyle{aer}
\bibliography{references.bib}

\appendix
\section{Proofs}
\subsection{Proof of Theorem \ref{thm:existence}}
\begin{proof}[Proof of Theorem \ref{thm:existence}]
The principal can always choose to offer the pooling contract $\left(q_H^*,C_H(q_H^*)\right)$ to all agents in every period. This induces no information revelation and leads to a profit of $\pi_H(q_H^*)$. If there is no mechanism that generates a profit greater than $\left(q_H^*,C_H(q_H^*)\right)$, we are done as this pooling solution is optimal.

Now assume that there is a mechanism that yields more than $\pi_{H}(q_{H}^{\ast }).$ Let $\pi ^{\ast }$ be the supremum of the profit over all such mechanisms. Clearly, $\pi^*$ is bounded by the commitment profit. For every $n\in \mathbb{N},$ take an implementable allocation 
\[
\mathcal{C}^{n}:=    \left\{\left(q_{t,n}^P,x_{t,n}^P\right), \left(q^S_{\tau,n}(t),x^S_{\tau,n}(t)\right), r_t\right\}_{t\geq 0, \tau\geq t},
\]
such that 
\[
\Pi \left( \mathcal{C}^{n}\right) >\pi ^{\ast }-n^{-1}.
\]

By assumption, $q_{t,n}^{P},q^S_{\tau,n}(t)\in [0,\bar{q}]$ and $r_n\in [0,1]$, and one can without loss set $x_{t,n}^{P},x^S_{\tau,n}(t)\leq C_H(\bar{q})$ (see Lemma \ref{lem:frontloading} and Lemma \ref{lem:poolingquantity}). Thus the elements in $\mathcal{C}^n$ are uniformly bounded. Taking a subsequence if necessary, $\mathcal{C}^{n}$ converges to an allocation $\mathcal{C}^{\ast }$ in the product topology. As the profit $\Pi \left(
\cdot \right) $ is continuous in product topology,%
\[
\Pi \left( \mathcal{C}^{\ast }\right) =\lim \Pi \left( \mathcal{C}%
^{n}\right) \geq \pi ^{\ast },
\]
and hence $\Pi \left( \mathcal{C}^{\ast }\right) =\pi ^{\ast }.$ Finally, because all constraints in program \eqref{eq:principal} are weak inequalities, they are preserved in the
limit. So $\mathcal{C}^{\ast }$ is implementable and hence optimal.
\end{proof}

\subsection{Proof of Theorem \ref{thm:structure}}
In this section, we assume that  the maximal profit is strictly greater than $\pi_H(q^*_H)$ (thus every optimal mechanism is revealing), and that the discount factor $\delta$ is greater than some threshold $\bar{\delta}$ which we shall define later.

For any implementable allocation $\left\{\left(q_t^P,x_t^P\right), \left(q^S_\tau(t),x^S_\tau(t)\right), r_t\right\}_{t\geq 0, \tau\geq t}$, let $\mathcal{R}$ denote the set of periods in which there is revelation, i.e., $r_t>0$. Recall that $U_{L,t}^S(T)$, defined in \eqref{eq:ULtT}, is the continuation payoff of a low type at $t$ if he revealed in period $T\leq t$, and $$U^S_{H,t} : = \sup_{\mathbb{T}\geq t}\left( 1-\delta \right)
\sum_{\tau =t}^{\mathbb{T}}\delta ^{\tau -t}\left[ x_{\tau }^{S}\left(
t\right) -C_H\left(q_{\tau }^{S}\left( t\right) \right)
\right]$$ is the high type's payoff is he (falsely) reveals at $t$.
\begin{claim}\label{cl:frontloading}
    Consider an implementable allocation $\left\{\left(q_t^P,x_t^P\right), \left(q^S_\tau(t),x^S_\tau(t)\right), r_t\right\}_{t\geq 0, \tau\geq t}$. 
    Define
    \begin{equation*}
        \Delta \tilde{X}_t^S \equiv \frac{U_{H,t}^S}{1-\delta},
    \end{equation*}
    and 
    \begin{equation}\label{eq:frontloadingconstruction}
        \tilde{x}_\tau^S(t) \equiv \begin{cases}
        C_H(q_\tau^S(t)) + \Delta \tilde{X}_t^S, &\text{ if } \tau = t,\\
        C_H(q_\tau^S(t)), &\text{ if } t< \tau < t+s\\
        C_L(q_\tau^S(t)) + \beta \Delta C (q_\tau^S(t)), &\text{ if } \tau = t+s\\
        C_L (q_\tau^S(t)), &\text{ if } \tau > t+s\\
        \end{cases}
    \end{equation}
    where $s\geq 0$ and $\beta\in (0,1]$ are determined by
    \begin{equation*}
    U_{L,t}^S(t) = (1-\delta)\Delta \tilde{X}_t^S + (1-\delta)\sum_{\tau = t}^{t+s-1}\delta^{\tau-t}\Delta C(q_\tau^S(t)) + (1-\delta)\delta^{s}\beta\Delta C(q_{t+s}^S(t)).
\end{equation*}
Then, $\left\{\left(q_t^P,x_t^P\right), \left(q^S_\tau(t),\tilde{x}^S_\tau(t)\right), r_t\right\}_{t\geq 0, \tau\geq t}$ is another implementable allocation that generates the same profit for the principal. Moreover, if $x_\tau^S(t) \neq  \tilde{x}_\tau^S(t)$ for some $\tau\geq t$, then \eqref{eq:nonreneging} is not always binding under the alternative allocation.
\end{claim}

\begin{proof}
    By construction, the alternative allocation delivers the same total surplus (because quantities are unchanged) and same payoffs to both types as the original allocation does, so the principal's profit, as well as \eqref{eq:IC-H}, \eqref{eq:IC-L}, \eqref{eq:IR-H}, and \eqref{eq:IR-L}, is unaltered. 
    
    We now argue that the alternative allocation (weakly) relaxes \eqref{eq:nonreneging}. This is true if for all $T\geq t$,
    \begin{equation}\label{eq:frontloadingverify}
        (1-\delta)\sum_{\tau=t}^{T}\delta^{\tau-t}\tilde{x}_\tau^S(t)\geq (1-\delta)\sum_{\tau=t}^{T}\delta^{\tau-t}x_\tau^S(t).
    \end{equation}
    That is, if at any $T\geq t$ the principal has paid more (i.e., owes less) to the low-type agent under the alternative allocation than she does under the original allocation.
    Suppose toward a contradiction that there is $T^*\geq t$ such that $(1-\delta)\sum_{\tau=t}^{T}\delta^{\tau-t}\tilde{x}_\tau^S(t)< (1-\delta)\sum_{\tau=t}^{T}\delta^{\tau-t}x_\tau^S(t)$. 
    \begin{itemize}
        \item If $T^*\leq t+s$, then under the original allocation, the high type can get a payoff greater than $(1-\delta)\Delta\tilde{X}_t^S$ by falsely revealing at $t$ and quitting after period $T^*$, a contradiction to the definition of $U_{H,t}^S$.
        \item If $T^*> t+s$, because $(1-\delta)\sum_{\tau=t}^{\infty}\delta^{\tau-t}\tilde{x}_\tau^S(t)= (1-\delta)\sum_{\tau=t}^{\infty}\delta^{\tau-t}x_\tau^S(t)$ by construction, we have $(1-\delta)\sum_{\tau=T^*+1}^{\infty}\delta^{\tau-t}\tilde{x}_\tau^S(t)> (1-\delta)\sum_{\tau=T^*+1}^{\infty}\delta^{\tau-t}x_\tau^S(t)$. But then,
        \begin{align*}
            U_{L,T^*+1}(t) &= \left( 1-\delta \right) \sum_{\tau =T^*+1}^{%
\mathbb{\infty }}\delta ^{\tau -T^*-1}\left[ x_{\tau}^{S}(t)-C_L \left(q^S_\tau(t)\right) \right]\\
&< \left( 1-\delta \right) \sum_{\tau =T^*+1}^{%
\mathbb{\infty }}\delta ^{\tau -T^*-1}\left[ \tilde{x}_{\tau}^{S}(t)-C_L \left(q^S_\tau(t)\right) \right] \\
&= 0,
        \end{align*}
a contradiction \eqref{eq:IR-L} at $T^*+1$.
    \end{itemize}

Finally, if $x_\tau^S(t) \neq \tilde{x}_\tau^S(t)$ for some $\tau\geq t$, let $\tilde{T}$ be the first period this occurs, so that $x_\tau^S(t) = \tilde{x}_\tau^S(t)$ for $\tau< \tilde{T}$ and $x_{\tilde{T}}^S(t) < \tilde{x}_{\tilde{T}}^S(t)$. Condition \eqref{eq:frontloadingverify} then holds strictly for $T=\tilde{T}$, implying that \eqref{eq:nonreneging} is slack at $\tilde{T}+1$ under the alternative allocation.
\end{proof}

\begin{claim}\label{cl:nolastperiod}
    There exists $\hat{\delta}\in \left( 0,1\right) $ such that if 
$\delta >$ $\hat{\delta}$, then any implementable allocation with $\mathcal{R}\neq \emptyset$ satisfies $\sup \mathcal{R} = \infty$, i.e., there cannot be a last revelation period. 
\end{claim}
\begin{proof}
    Suppose that $\sup\mathcal{R} = T<\infty$. Because there is no information revelation after $T$, \eqref{eq:nonreneging} requires that full exploitation of both types happens after $T$, that is, for all $t>T$ we have $q_t^P = q_H^*$, $x_t^P = C_H(q_H^*)$, $q_{t}^S(T) =q_L^*$, and $x_{t}^S(T) =C_L(q_L^*)$.

    To incentivize a low type to reveal at $T$, we need
    \begin{equation*}
        (1-\delta)[x_T^S(T)-C_L(q_T^S(T))]\geq (1-\delta)[x_T^P-C_L(x_T^P)] + \delta \Delta C(q_H^*),
    \end{equation*}
    where the LHS is the low type's payoff from revealing at $T$ and the RHS is his payoff from taking the pooling contract forever. On the other hand, to make sure that the high type does not want to falsely reveal at $T$ (and then quit right after), we need
    \begin{equation*}
        (1-\delta)[x_T^S(T)-C_H(q_T^S(T))]\leq (1-\delta)[x_T^P-C_H(x_T^P)].
    \end{equation*}
    Combining these two inequalities, we have
    \begin{align*}
        (1-\delta)[x_T^S(T)-C_L(q_T^S(T))]&\geq (1-\delta)[x_T^P-C_L(x_T^P)] + \delta \Delta C(q_H^*)\\
        &= (1-\delta)[x_T^P-C_H(x_T^P)] + (1-\delta)\Delta C(x_T^P) + \delta \Delta C(q_H^*)\\
        &\geq (1-\delta)[x_T^S(T)-C_H(q_T^S(T))] + (1-\delta)\Delta C(x_T^P) + \delta \Delta C(q_H^*)\\
        &= (1-\delta)[x_T^S(T)-C_L(q_T^S(T))] + (1-\delta)\Delta C(x_T^P) + \delta \Delta C(q_H^*) \\ 
        &\quad\ -(1-\delta) \Delta C(q_T^S(T)),
    \end{align*}
    which implies that
    \begin{equation*}
        (1-\delta)\Delta C(q_T^S(T)) \geq \delta \Delta C(q_H^*).
    \end{equation*}
    But because $\Delta C(q_T^S(T))\leq \Delta C(\bar{q})$, the above inequality cannot hold if $\frac{1-\delta}{\delta}<\frac{\Delta C(q_H^*)}{\Delta C (\bar{q})}$, that is, if $\delta >\frac{\Delta C(\bar{q})}{\Delta C (\bar{q})+\Delta C(q_H^*)}:=\hat{\delta}$.
\end{proof}

We assume now that $\delta>\hat{\delta}$. Later we will further refine this lower bound (see Claim \ref{cl:bardelta}).
\begin{claim}\label{cl:lowtypeindifference}
    In any optimal mechanism, the low-type agent is indifferent between revealing or not at every $t\in \mathcal{R}$.
\end{claim}
\begin{proof}
    This follows immediately from Claim \ref{cl:nolastperiod} and the presumption that the any optimal mechanism is revealing.
\end{proof}

\begin{claim}\label{cl:poolingpayment}
    In any optimal mechanism, the high type's period-$0$ payoff satisfies $U^P_{H,0}=0$. There exists an optimal mechanism with $U^P_{H,t}=0$ for all $t\geq 0$, hence $x_t^P=C_H(q_t^P)$ in every period. Moreover, if $U^P_{H,t}>0$ at some $t$, then there exists an optimal allocation under which \eqref{eq:nonreneging} is not always binding.
\end{claim}
\begin{proof}
    Clearly $U_{H,0}=0$ for otherwise payments at $t=0$ can be uniformly decreased by $\epsilon$. 

    Take any optimal allocation $\left\{\left(q_t^P,x_t^P\right), \left(q^S_\tau(t),x^S_\tau(t)\right), r_t\right\}_{t\geq 0, \tau\geq t}$. Consider an alternative allocation $\left\{\left(q_t^P,\tilde{x}_t^P\right), \left(q^S_\tau(t),\tilde{x}^S_\tau(t)\right), r_t\right\}_{t\geq 0, \tau\geq t}$ such that the pooling payment satisfies $\tilde{x}_t^P=C_H(q_t^P)$ for every $t$ and the separating payment $\tilde{x}_\tau^S(t)$ is defined by \eqref{eq:frontloadingconstruction} for every $t\in \mathcal{R}$ and $\tau\geq t$ with $\Delta \tilde{X}_t^S =0$ and $s$ and $\beta$ being determined by
    $$
    (1-\delta)\sum_{\tau=t}^{\infty}\delta^{\tau-t}\Delta C(q_\tau^P) = (1-\delta)\sum_{\tau = t}^{t+s-1}\delta^{\tau-t}\Delta C(q_\tau^S(t)) + (1-\delta)\delta^{s}\beta\Delta C(q_{t+s}^S(t)).
    $$
    Let $\tilde{U}_{\theta,t}^S$ and $\tilde{U}_{\theta,t}^P$ denote the players' payoffs under this alternative allocation. By construction, $\tilde{U}_{H,t}^P=\tilde{U}_{H,t}^S = 0$ for every $t$ and $\tilde{U}_{L,t}^P =\tilde{U}_{L,t}^S(t) = (1-\delta)\sum_{\tau=t}^{\infty}\delta^{\tau-t}\Delta C(q_\tau^P)$ for every $t\in\mathcal{R}$, so \eqref{eq:IC-H}, \eqref{eq:IC-L}, \eqref{eq:IR-H}, and \eqref{eq:IR-L} are satisfied.

    We now show that the alternative allocation (weakly) relaxes \eqref{eq:nonreneging}. By the argument in the proof of Lemma \ref{cl:frontloading}, the variation of the separating payments relaxes \eqref{eq:nonreneging}. For the variation of the pooling payments, note that for any $T\geq 0$ we must have
    \begin{equation}\label{eq:frontloadingverifypooling}
        \sum_{\tau=0}^{T}\delta^\tau \tilde{x}_\tau^P\geq \sum_{\tau=0}^{T}\delta^\tau x_\tau^P,
    \end{equation}
    for otherwise if it were the case that $\sum_{\tau=0}^{T}\delta^\tau x_\tau^P > \sum_{\tau=0}^{T}\delta^\tau \tilde{x}_\tau^P$, then under the original allocation the high type would obtain a strictly positive payoff by working until period $T$ and then quitting, a contradiction to $U_{H,0} = 0$. Therefore, at any $T\geq 0$ the principal has paid more (i.e., owes less) to the high-type agent under the alternative allocation than she does under the original allocation, which further relaxes \eqref{eq:nonreneging}.

    Next we argue that the principal's time-$0$ profit keeps the same under the alternative allocation. First, $U_{H,0}^P = 0 = \tilde{U}_{H,0}^P$. Now consider a low-type agent who reveals in period $t$ under the original allocation. Because the original allocation is optimal, Claim \ref{cl:lowtypeindifference} implies the low type's time-$0$ payoff can be attained by taking the pooling contract forever, that is,
    \begin{equation*}
        U_{L,0}^P = U_{H,0}^P + \sum_{\tau =0}^{\infty}\delta^\tau\Delta  C(q_\tau^P).
    \end{equation*}
    Meanwhile, by construction of the alternative allocation, the low type is indifferent between when to reveal too, so the low type's time-$0$ payoff is
    \begin{equation*}
        \tilde{U}_{L,0}^P = \sum_{\tau =0}^{\infty}\delta^\tau\Delta C(q_\tau^P).
    \end{equation*}
    Recall that $U_{H,0}^P=0$, so we have $U_{L,0}^P = \tilde{U}_{L,0}^P$. As quantities (thus surplus) are unchanged, the principal's time-$0$ profit remains the same. Therefore, the alternative allocation is also optimal.

    Finally, if $U_{H,t}>0$ for some $t$, then $x_\tau^P\neq \tilde{x}_\tau^P$ for some $\tau$. Let $\tilde{T}$ be the first period this occurs, so that $x_\tau^P = \tilde{x}_\tau^P$ for $\tau< \tilde{T}$ and $x_{\tilde{T}}^P < \tilde{x}_{\tilde{T}}^P$. Condition \eqref{eq:frontloadingverifypooling} then holds strictly for $T=\tilde{T}$, implying that \eqref{eq:nonreneging} is slack at $\tilde{T}+1$ under the alternative optimal allocation.
\end{proof}

\begin{definition}
    Consider a low-type agent who reveals in period $t$ and is offered $\{q_\tau^S(t),x_\tau^S(t)\}_{\tau\geq t}$. We say that the principal is \textbf{in debt} with this agent in period $\tau$ if the agent's continuation payoff is strictly positive. Any such period is called a \textbf{debt period.} 
\end{definition}

\begin{claim}\label{cl:bardelta}
    There exists $\bar{\delta}\in(\hat{\delta},1)$ such that if $\delta>\bar{\delta}$, any optimal allocation with separating payments given by \eqref{eq:frontloadingconstruction} satisfies $s\geq 1$ for every $t\in \mathcal{R}$. That is, it takes at least two debt periods to reward a low type after his revelation.
\end{claim}

\begin{proof}
    We prove a stronger claim: For every $N\in \mathbb{N}$, there exists $\delta^*\in (\hat{\delta},1)$ such that if $\delta>\delta^*$ and revelation happens in period $t$, then the reward of the low type takes at least $N$ periods.

    Assume toward a contradiction that we can find a sequence of optimal allocations associated with $\delta _{n} \rightarrow 1$ and $M>0$ such that it takes less than $M$ periods to reward the low type in some revealing period with discount factor $\delta _{n}.$ Then by revealing the low type can get at most $\left( 1-\delta_n \right) M \Delta C \left( \bar{q}\right) $ which converges to zero as $\delta_{n}\rightarrow 1.$ Because the low type is indifferent between revealing or not (Claim \ref{cl:lowtypeindifference}), this implies that the low type's payoff from taking the pooling contract converges to zero. Hence, for every $\epsilon >0$, the principal's continuation payoff is bounded above by $(1-\alpha_0)\epsilon + \alpha_0\pi_L(q_L^*)$. But the principal is always able to guarantee herself a continuation payoff of $\pi_H(q_H^*)$, which is greater than $\alpha_0\pi_L(q_L^*)$ by Assumption \ref{ass:1}, a contradiction.
\end{proof}

We assume for the reminder of the proof that $\delta>\bar{\delta}$.

\begin{claim}\label{cl:hightypedistortion}
    In any optimal mechanism, $q_t^P\leq q_H^*$ for every $t$, and $q_t^P< q_H^*$ for every $t\in\mathcal{R}$.
\end{claim}

\begin{proof}
    Take any optimal mechanism and, if necessary, apply the variation of payments as in Claims \ref{cl:frontloading} and \ref{cl:poolingpayment}. If $q_t^P>q_H^*$ at some $t$, then decreasing $q_t^P$ to $q_H^*$ and $x_t^P$ to $C_H(q_H^*)$ will increase the surplus at $t$ and decrease the information rent to any low-type agents who revealed before $t$. Moreover, if $r_t>0$ and $q_t^P=q_H^*$, further decreasing $q_t^P$ by $\epsilon$ has no first-order effect on surplus but can reduce the information rent to the low-type agents who reveal at $t$ by $r_t(1-\delta)\Delta C'(q_H^*)\epsilon$.
\end{proof}

\begin{claim}\label{cl:gradualrevelation}
    In any optimal mechanism, $\mathcal{R} = \mathbb{Z}_+$. That is, information revelation is gradual.
\end{claim}
\begin{proof}
    Take any optimal mechanism and, if necessary, apply the variation of payments as in Claims \ref{cl:frontloading} and \ref{cl:poolingpayment}. 
    
    We first argue that $r_0>0$. If not, let $\tilde{t} = \min\mathcal{R}$. Note that the principal's profit is at most $\pi_H(q_H^*)$ at each $t\in \{0,...,\tilde{t}-1\}$. The principal's total time-$0$ profit is given by $(1-\delta)\left[\sum_{\tau=0}^{\tilde{t}-1}\delta^\tau\pi_H(q_\tau^P) + \delta^{\tilde{t}}\Pi_{\tilde{t}}\right]$, which is presumed to be greater than $\pi_H(q_H^*)$, so we must have $\Pi_{\tilde{t}} > \Pi_0$. But then, the principal is strictly better off by skipping the first $\tilde{t}$ periods and starting with the contracts from time-$\tilde{t}$ onward.

    Next we show that if $r_t>0$ then $r_{t+1}>0$. Suppose toward a contradiction that $r_{t}>0$ but $r_{t+1}=0.$ Let $\hat{t}$ be the smallest integer greater than $t$ such that $r_{\hat{t}}>0$. Such $\hat{t}$ exists because $\sup\mathcal{R} = \infty$ by Claim \ref{cl:nolastperiod}. Hence $r_{\tau }=0$ for every $\tau \in \{t+1,...,\hat{t}-1\}.$  Since $\delta > \bar{\delta}$, by Claim \ref{cl:bardelta} the principal must be paying debt at period $t+1$. But then, from the principal's perspective, her incentive to renege at $t+1$ is strictly stronger than that at $\hat{t}$. Indeed, relative to reneging at $\hat{t}$, reneging at $t+1$ saves the principal from paying her time-$t+1$ debt while she loses nothing as no one reveals between $t+1$ and $\hat{t}$. This implies that \eqref{eq:nonreneging} does not bind at period $\hat{t}$, for otherwise the principal would find it optimal to renege at $t+1$. Therefore we can implement a variation of the mechanism to induce a fraction $\phi\in \left( 0,1\right) $ of the agents that were supposed to reveal at $\hat{t}$ to reveal at $\hat{t}-1$. This can be achieved by offering them a sequence of separating contracts $\{\tilde{q}_{\tau}^S(\hat{t}-1),\tilde{x}_{\tau}^S(\hat{t}-1)\}_{\tau\geq \hat{t}-1}$, such that $\tilde{q}_{\hat{t}-1}^S(\hat{t}-1)= q_L^*$, $\tilde{q}_{\tau}^S(\hat{t}-1)= q_{\tau}^S(\hat{t})$ for $\tau\geq \hat{t}$, and the payments are given by \eqref{eq:frontloadingconstruction}. This is incentive-compatible, keeps the agent indifferent, and increases the surplus, because before the variation this fraction of low type agents only produce $q_{\hat{t}-1}^P$ at $\hat{t}-1$, which is strictly less than $q_L^*$ by Claim \ref{cl:hightypedistortion}. Thus this is a profitable deviation and contradicts the optimality of the original mechanism.
\end{proof}

\begin{claim}\label{cl:principalindifference}
     In any optimal mechanism, \eqref{eq:nonreneging} is binding in every period $t\in\{1,2,...\}$. That is, the principal is always indifferent between reneging on past promises or not.
\end{claim}

\begin{proof}
Take any optimal mechanism, and suppose toward a contradiction that \eqref{eq:nonreneging} is slack in some period. Let $\hat{t}$ be the smallest period this happens. By Claim \ref{cl:gradualrevelation}, $r_{\hat{t}}>0$. As in the proof of Claim \ref{cl:gradualrevelation}, one can construct a sequence of separating contracts $\{\tilde{q}_{\tau}^S(\hat{t}-1),\tilde{x}_{\tau}^S(\hat{t}-1)\}_{\tau\geq \hat{t}-1}$ to induce a fraction $\phi$ of the agents that were supposed to reveal at $\hat{t}$ to reveal at $\hat{t}-1$. This variation increases the principal's profit by the previous argument. Note, however, that the principal is already offering a sequence of separating contracts $\{{q}_{\tau}^S(\hat{t}-1),{x}_{\tau}^S(\hat{t}-1)\}_{\tau\geq \hat{t}-1}$ for revealing at $\hat{t}-1$, so this variation will result in a randomization between two separating offers starting $\hat{t}-1$, which is not allowed by the definition of a sequential revelation mechanism. To construct a feasible profitable deviation, we will replace $\{{q}_{\tau}^S(\hat{t}-1),{x}_{\tau}^S(\hat{t}-1)\}_{\tau\geq \hat{t}-1}$ with a ``convex combination" of $\{{q}_{\tau}^S(\hat{t}-1),{x}_{\tau}^S(\hat{t}-1)\}_{\tau\geq \hat{t}-1}$ and $\{\tilde{q}_{\tau}^S(\hat{t}-1),\tilde{x}_{\tau}^S(\hat{t}-1)\}_{\tau\geq \hat{t}-1}$. This single offer will be taken by all of the $r_{\hat{t}-1} + \phi r_{\hat{t}}$ agents revealing at $
\hat{t}-1$.

Specifically, let $(\chi_{\hat{t}-1},\chi_{\hat{t}},\chi_{\hat{t}+1},...)$ and $(\tilde{\chi}_{\hat{t}-1},\tilde{\chi}_{\hat{t}},\tilde{\chi}_{\hat{t}+1},...)$ be the sequences of rents the low type gets from $\{{q}_{\tau}^S(\hat{t}-1),{x}_{\tau}^S(\hat{t}-1)\}_{\tau\geq \hat{t}-1}$ and $\{\tilde{q}_{\tau}^S(\hat{t}-1),\tilde{x}_{\tau}^S(\hat{t}-1)\}_{\tau\geq \hat{t}-1}$, respectively. Note that the discounted sum of each sequence of rents is the same. Define $$\hat{\chi}_{\tau} : = \frac{r_{\hat{t}-1}}{r_{\hat{t}-1}+\phi r_{\hat{t}}}\chi_{\tau} + \frac{\phi r_{\hat{t}}}{r_{\hat{t}-1}+\phi r_{\hat{t}}}\tilde{\chi}_{\tau}.$$
Consider a new sequence of separating contracts $\{\hat{q}_{\tau}^S(\hat{t}-1),\hat{x}_{\tau}^S(\hat{t}-1)\}_{\tau\geq \hat{t}-1}$ such that 
\begin{equation*}
    \hat{q}_{\tau}^S(\hat{t}-1):=\begin{cases}
        q_L^*, &\text{ if }\hat{\chi}_{\tau}\leq \Delta C(q_L^*)\\
        \Delta C^{-1}(\hat{\chi}_{\tau}), &\text{ if }\hat{\chi}_{\tau}>\Delta C(q_L^*)
    \end{cases},
\end{equation*}
and
\begin{equation*}
    \hat{x}_{\tau}^S(\hat{t}-1):=
    \begin{cases}
        C_L(q_L^*) + \hat{\chi}_{\tau}, &\text{ if }\hat{\chi}_{\tau}\leq \Delta C(q_L^*)\\
        C_H(\hat{q}_{\tau}^S(\hat{t}-1)), &\text{ if }\hat{\chi}_{\tau}>\Delta C(q_L^*)
    \end{cases}.
\end{equation*}
By construction, the low type's payoffs from this new sequence of contracts are $(\hat{\chi}_{\hat{t}-1},\hat{\chi}_{\hat{t}},\hat{\chi}_{\hat{t}+1},...)$ whose discounted sum is the same as the other two rent sequences. 

Now we show that this new sequence of separating contracts increases the principal profit against the randomization of the two contract sequences we started with. Note that the principal's profit in period $\tau$ is given by
\begin{equation*}
    \Psi(\hat{\chi}_{\tau}) =\begin{cases}
        \pi_L(q_L^*) - \hat{\chi}_\tau,&\text{ if }\hat{\chi}_{\tau}\leq \Delta C(q_L^*)\\
        \pi_H(\hat{q}_{\tau}^S(\hat{t}-1)), &\text{ if }\hat{\chi}_{\tau}>\Delta C(q_L^*)
    \end{cases}.
\end{equation*}
Using the concavity of $\pi_H(\cdot)$ and $\Delta C(\cdot)$, one can verify that $\Psi(\cdot)$ is concave. This implies that for each $\tau\geq \hat{t}-1$,
\begin{equation*}
    \Psi(\hat{\chi}_{\tau}) \geq \frac{r_{\hat{t}-1}}{r_{\hat{t}-1}+\phi r_{\hat{t}}}\Psi(\chi_{\tau}) + \frac{\phi r_{\hat{t}}}{r_{\hat{t}-1}+\phi r_{\hat{t}}}\Psi(\tilde{\chi}_{\tau}).
\end{equation*}
So we have a profitable deviation.
\end{proof}

\begin{corollary}\label{cor:frontloading}
    In any optimal mechanism, the separating payments are given by \eqref{eq:frontloadingconstruction} and the pooling payments satisfy $x_t^P = C_H(q_t^P)$ for every $t$.
\end{corollary}

\begin{proof}
    This follows immediately from Claims \ref{cl:frontloading}, \ref{cl:poolingpayment}, and \ref{cl:principalindifference}.
\end{proof}

 \begin{proof}[\bf Proof of Lemma \ref{lem:frontloading}]
     This follows immediately from Corollary \ref{cor:frontloading}, Claim \ref{cl:poolingpayment}, and Claim \ref{cl:bardelta}.
 \end{proof}

 \begin{claim}\label{cl:eventualrevelation}
    In any optimal mechanism, $\sum_{t=0}^{\infty}r_t =\alpha_0$. That is, information revelation is eventual.
\end{claim}
\begin{proof}
    Take any optimal mechanism. Recall that $R_t : = \sum_{\tau<t} r_\tau$. Suppose toward a contradiction that $\lim_t R_t = \bar{R} < \alpha_0$. Let $V^P$ be the principal's profit from an agent who takes the pooling contract forever, and $V^S$ be the principal's average profit from the agents who reveal at some $t$. Remember we assume  $\pi_H(q_H^*)< \Pi_0 = (1-\bar{R})V^P + \bar{R} V^S$.   By Claim \ref{cl:hightypedistortion}, $V^P<\pi_H(q_H^*)$. This implies that 
    $$V^S > \pi_H(q_H^*) > V^P.$$

    Now consider an alternative allocation where all contracts remain the same while the measure of agents revealing at $t$ is scaled up to $\hat{r}_t = \frac{\alpha_0}{\bar{R}} r_t$ for every $t$. Clearly this alternative allocation still satisfies \eqref{eq:IC-H}, \eqref{eq:IC-L}, \eqref{eq:IR-H}, and \eqref{eq:IR-L}, because all contracts are unaltered. Moreover, the principal's profit under the alternative allocation is $\alpha_0 V^S + (1-\alpha_0)V^P$, which is greater than $\Pi_0$. This renders a valid profitable deviation and contradicts the optimality of the original mechanism, if we can show that \eqref{eq:nonreneging} remains satisfied.

    To check \eqref{eq:nonreneging}, let us fix any $t$. Claim \ref{cl:principalindifference} implies that under the original optimal allocation the principal is indifferent between reneging at $t$ or reneging at $t+1$. Consider three types of agents:
    \begin{itemize}
        \item Those who are supposed to reveal at $t$, with measure $r_t$. If the principal reneges at $t$, she gets $r_t\pi_H(q_H^*)$ from them; if the principal reneges at $t+1$, she gets $r_t[(1-\delta)\pi_H(q_t^S(t))+\delta \pi_L(q_L^*)]$ from them. 
        \item Those who revealed prior to $t$, with measure $R_t$. Let $B_t^S(\tau)$ be the rent the principal is suppose to pay at $t$ to an agent who revealed at $\tau<t$. If the principal reneges at $t$, she gets $R_t\pi_L(q_L^*)$ from them; if the principal reneges at $t+1$, she gets $(1-\delta)\sum_{\tau<t}r_\tau [v(q_t^S(\tau))-B_t^S(\tau)] + \delta R_t \pi_L(q_L^*)$.

        \item Those who are supposed to take the pooling contracts at $t$, with measure $1-R_{t+1}$. If the principal reneges at $t$, she gets $(1-R_{t+1})\pi_H(q_H^*)$ from them; if the principal reneges at $t+1$, she gets $(1-R_{t+1})[(1-\delta)\pi_H(q_t^P)+\delta\pi_H(q_H^*)]$.
    \end{itemize}

    Because the principal weakly prefers reneging at $t+1$ to reneging at $t$, we have
    \begin{align*}
        &r_t[(1-\delta)\pi_H(q_t^S(t))+\delta \pi_L(q_L^*)-\pi_H(q_H^*)] \\
        &+ (1-\delta)\sum_{\tau<t}r_\tau [v(q_t^S(\tau))-B_t^S(\tau)-\pi_L(q_L^*)]  \\
        &+ (1-R_{t+1})(1-\delta)[\pi_H(q_t^P)-\pi_H(q_H^*)] \\
        &\geq 0.
    \end{align*}
As $\pi_H(q_t^P)-\pi_H(q_H^*)\leq 0$, the above inequality continues to hold if every $r_\tau$ is scaled up the same factor (greater than $1$). This means that \eqref{eq:nonreneging} is still satisfied at $t$. Because $t$ is arbitrary, we are done.
\end{proof}

\begin{proof}[\bf Proof of Lemma \ref{lem:info}]
    This follows immediately from Claims \ref{cl:gradualrevelation} and \ref{cl:eventualrevelation}.
\end{proof}

\begin{claim}\label{cl:poolingquantityconverge}
    In any optimal mechanism, $\lim_{t\to\infty}q_t^P = q_H^*$.
\end{claim}

\begin{proof}
    Notice that $r_t\to 0$. By Claim \ref{cl:principalindifference}, at each $t$ the principal is indifferent between reneging at $t$ and reneging at $t+1$. By reneging today, the principal can save the information rents to those whom she is in debt with and, more importantly, let the high type to produce $q_H^*$ today rather than $q_t^P$. Therefore, relative to reneging at $t+1$, the principal's benefit of reneging at $t$ is bounded below by $(1-\delta)(1-\alpha_0)[\pi_H(q_H^*)- \pi_H(q_t^P)]$. On the other hand, the cost of reneging at $t$ comes from not being able to exploit the agents who are supposed to reveal at $t$. This is proportional to $r_t$ and converges to $0$ as $t\to\infty$. Therefore, for the principal to be indifferent between reneging at $t$ or $t+1$, we must have $(1-\delta)(1-\alpha_0)[\pi_H(q_H^*)- \pi_H(q_t^P)]\to 0$, which implies $\lim_t q_t^P =  q_H^*$.
\end{proof}

\begin{proof}[\bf Proof of Lemma \ref{lem:poolingquantity}]
    This follows immediately from Claims \ref{cl:poolingpayment}, \ref{cl:hightypedistortion}, \ref{cl:gradualrevelation}, and \ref{cl:poolingquantityconverge}.
\end{proof}

\begin{proof}[\bf Proof of Lemma \ref{lem:separatingquantity}]
Take any optimal mechanism and fix any $t$. 

First we argue that $q_\tau^S(t)\geq q_L^*$ for every $\tau\geq t$. Suppose toward a contradiction that we can find $\tau\geq t$ such that $q_\tau^S(t)< q_L^*$. The principal can then increase $q_\tau^S(t)$ to $q_L^*$ and increase $x_\tau^S(t)$ by $C_L(q_L^*)-C_L(q_\tau^S(t))$. This keeps the payoff of the low type unchanged, relaxes the high type's incentive compatibility, and increases surplus, making it a profit deviation.

Next we show that $q_\tau^S(t)=q_L^*$ for all $\tau>t+s$, that is, when the principal is no longer in debt with the agents who reveals at $t$. Indeed, if $q_\tau^S(t)\neq q_L^*$, the principal can simply reset $q_\tau^S(t)$ to $q_L^*$ and $x_\tau^S(t)$ to $C_L(q_L^*)$. This increases surplus, keeps the low type's payoff unchanged, and maintains the high type's incentive compatibility,\footnote{Note that the payment after the debt periods is less than the high type's cost, so the high type (even after falsely revealing at $t$) would never stay in the separating contract till then.} making it a profitable deviation.

Now we prove that if $\tau$ and $\tau+1$ are debt periods then $q_{\tau}^S(t)>q_{\tau+1}^S(t)$. 
\begin{itemize}
    \item Suppose toward a contradiction that $q_{\tau}^S(t)<q_{\tau+1}^S(t)$. By Lemma \ref{lem:frontloading}, $x_{\tau}^S(t) = C_H(q_{\tau}^S(t))$ and $x_{\tau+1}^S(t) = C_H(q_{\tau+1}^S(t))$.\footnote{This is true even if $\tau+1$ is the last debt period. To see this, note that $q_{\tau+1}^S(t)>q_{\tau}^S(t)\geq q_L^*$ by the contradiction assumption. If it were the case that $x_{\tau+1}^S(t) < C_H(q_{\tau+1}^S(t))$, then the principal could decrease $q_{\tau+1}^S(t)$ and $x_{\tau+1}^S(t)$ at the same time to keep the low type's payoff unchanged and increase surplus. Importantly, this variation would not affect the high type's incentive compatibility because the high type would never falsely reveal at $t$ and stay in the contract until $\tau+1$. So this would be a profitable deviation.} Now increase $q_{\tau}^S(t)$ by $\epsilon$, decrease $q_{\tau+1}^S(t)$ by $\frac{\Delta C'(q_{\tau}^S(t))}{\Delta C'(q_{\tau+1}^S(t))}\delta^{-1}\epsilon$, and adjust payments accordingly. By construction, this variation keeps the low type's payoff unchanged and frontloads some rent from $\tau+1$ to $\tau$ (thus relaxing \eqref{eq:nonreneging} at $\tau+1$). Moreover, the change of surplus is given by
    \begin{align*}
        (1-\delta)\epsilon\left[\pi_H'(q_{\tau}^S(t)) - \pi_H'(q_{\tau+1}^S(t))\frac{\Delta C'(q_{\tau}^S(t))}{\Delta C'(q_{\tau+1}^S(t))}\right], 
    \end{align*}
    which is positive because \textit{i) $q_L^*\leq q_{\tau}^S(t)<q_{\tau+1}^S(t)$ and $\pi_H$ is strictly concave so that $0\geq \pi_H'(q_{\tau}^S(t)) > \pi_H'(q_{\tau+1}^S(t))$}; and \textit{ii)} $\Delta C$ is concave so that $\frac{\Delta C'(q_{\tau}^S(t))}{\Delta C'(q_{\tau+1}^S(t))}\geq 1$. So this is a profitable deviation. 
\item Suppose toward a contradiction that $q_{\tau}^S(t)=q_{\tau+1}^S(t)$. We can apply the same variation as above, which keeps the agents' payoffs and has no first-order effect on surplus this time. But it still strictly relaxes \eqref{eq:nonreneging} at $t+1$, thus a further variation based on Claim \ref{cl:principalindifference}'s proof will render a profitable deviation.
    \end{itemize}

Therefore, $x_\tau^S(t)$ must be strictly decreasing in $\tau$ over the debt periods and stay constant at $q_L^*$ after the debt periods. This establishes Lemma \ref{lem:separatingquantity}.
\end{proof}

\begin{proof}[\bf Proof of Theorem \ref{thm:structure}]
    This theorem follows directly from Lemmas \ref{lem:frontloading} through \ref{lem:separatingquantity}.
\end{proof}

\subsection{Proof of Theorem \ref{thm:nsc} and Proposition \ref{prop:typedistance}}\label{app:thm3}
\begin{proof}[\bf Proof of Part 1 of Theorem \ref{thm:nsc}]
Suppose that \eqref{eq:nsconditionnonrevealing} holds for all $\tilde{q}_L \in [q_L^*,\bar{q}]$, that is,
\begin{align*}
        \pi_L(q_L^*) - \pi_H(q_H^*) < \ln \left(\frac{\Delta C(\tilde{q}_L)}{\Delta C(\tilde{q}_L)-\Delta C({q}_H^*)}\right) \left[\pi_L(q_L^*)-\pi_H(\tilde{q}_L)\right].
\end{align*}
Let us first prove that the best pooling mechanism is optimal when $\delta$ is sufficiently close to $1$. Suppose (toward a contradiction) that for any $\epsilon>0$, there is some $\delta>1-\epsilon$ such that every optimal mechanism is revealing. We will construct a $\tilde{q}$ that violates the above inequality.

Let us first fix a small $\epsilon$ and such a  $\delta>1-\epsilon$. Take any optimal mechanism. For every $t$, let $\{z(t),...,t-1\}$ be the set of periods that the principal is still in debt with the agents revealing in these periods. By Lemma \ref{lem:frontloading}, if an agent reveals in $\tau\leq t$, the principal's (flow) payoff from this agent in period $t$ is 
\begin{equation*}
    \pi_t^S(\tau) =\begin{cases}
        \pi_H(q_t^S(\tau)),&\text{ if } x_t^S(\tau) = C_H(q_t^S(\tau))\\
        \pi_H(q_t^S(\tau)) + (1-\beta) \Delta C(q_t^S(\tau)),&\text{ if } x_t^S(\tau) = C_L(q_t^S(\tau))+\beta\Delta C(q_t^S(\tau))
    \end{cases}.
\end{equation*}

By Claim \ref{cl:principalindifference}, the principal is indifferent between reneging at $t$ or reneging at $t+1$. Following the analysis in the proof of Claim \ref{cl:eventualrevelation}, the principal's weak preference for reneging at $t+1$ relative to reneging at $t$ tells us that
\begin{align}\label{eq:NRwaiting}
        &r_t[(1-\delta)\pi_H(q_t^S(t))+\delta \pi_L(q_L^*)-\pi_H(q_H^*)] \nonumber\\
        \geq 
        &(1-\delta)\sum_{\tau=z(t)}^{t-1}r_\tau [\pi_L(q_L^*)- \pi_{t}^S(\tau)] + (1-R_{t+1})(1-\delta)[\pi_H(q_H^*)- \pi_H(q_t^P)].
    \end{align}
Because $\pi_H(q_H^*) \geq \pi_H(q_t^S(t))$ and $\pi_H(q_H^*) \geq \pi_H(q_t^P)$, the above inequality implies that 
\begin{align*}
        r_t\delta [\pi_L(q_L^*)-\pi_H(q_H^*)] 
        &\geq 
        (1-\delta)\sum_{\tau=z(t)}^{t-1}r_\tau [\pi_L(q_L^*)- \pi_{t}^S(\tau)].
    \end{align*}
Summing from any $\tilde{t}$ onward, we have
\begin{align}
        \sum_{t= \tilde{t}}^\infty r_t\delta [\pi_L(q_L^*)-\pi_H(q_H^*)] 
        &\geq (1-\delta)\sum_{t= \tilde{t}}^\infty
        \sum_{\tau=z(t)}^{t-1}r_\tau [\pi_L(q_L^*)- \pi_{t}^S(\tau)]\nonumber\\
        &\geq (1-\delta)\sum_{t=\tilde{t} }^{\infty}r_t\sum_{\tau = t+1}^{T(t)}\left[\pi_L(q_L^*) - \pi_\tau^S(t)\right]\nonumber\\
        &\geq (1-\delta)\sum_{t=\tilde{t} }^{\infty}r_t\sum_{\tau = t+1}^{T(t)-1}\left[\pi_L(q_L^*) - \pi_H(q_\tau^S(t))\right].\label{eq:rolling}
    \end{align}
The second line follows from changing the order of summation, dropping the terms from periods before $\tilde{t}$, and letting $T(t)$ be the final debt period for the agents revealing at $t$; the third line follows from the fact $\pi_\tau^S(t) = \pi_H(q_\tau^S(t))$ in all debt periods except the last one $T(t)$.

To further simply the RHS of \eqref{eq:rolling}, recall that the low types are indifferent between revealing or not in every period (Claim \ref{cl:lowtypeindifference}). And because $\lim_{t\to\infty}q_t^P = q_H^*$ (Claim \ref{cl:poolingquantityconverge}), we know that for $t$ sufficiently large, we have
\begin{equation*}
   (1-\delta)\sum_{\tau = t+1}^{T(t)-1}\delta^{\tau-t}\Delta C (q_\tau^S(t)) = \Delta C (q_H^*) + O(\epsilon).
\end{equation*}
Let $\tilde{q}_t\in [q_L^*,\bar{q}]$ be such that
\begin{equation*}
    (1-\delta)\sum_{\tau = t+1}^{T(t)-1}\delta^{\tau-t}\Delta C (\tilde{q}_t) = (1-\delta)\sum_{\tau = t+1}^{T(t)-1}\delta^{\tau-t}\Delta C (q_\tau^S(t)).
\end{equation*}
Because $\pi_H(\cdot)$ and $\Delta C (\cdot)$ are concave, we have 
\begin{equation*}
\sum_{\tau = t+1}^{T(t)-1} \pi_H(\tilde{q}_t) \geq \sum_{\tau = t+1}^{T(t)-1} \pi_H(q_\tau^S(t)).
\end{equation*}
Thus \eqref{eq:rolling} now becomes
\begin{align*}
        \sum_{t= \tilde{t}}^\infty r_t\delta [\pi_L(q_L^*)-\pi_H(q_H^*)] 
        &\geq (1-\delta)\sum_{t=\tilde{t} }^{\infty}r_t\sum_{\tau = t+1}^{T(t)-1}\left[\pi_L(q_L^*) - \pi_H(\tilde{q}_t)\right]\nonumber\\
        &= (1-\delta) \sum_{t=\tilde{t} }^{\infty}r_t [T(t)- t - 1]  \left[\pi_L(q_L^*) - \pi_H(\tilde{q}_t)\right];
    \end{align*}
that is, 
\begin{equation*}
    \sum_{t= \tilde{t}}^\infty r_t\left\{\delta [\pi_L(q_L^*)-\pi_H(q_H^*)] - (1-\delta) [T(t)- t - 1]  \left[\pi_L(q_L^*) - \pi_H(\tilde{q}_t)\right]\right\}\geq 0.
\end{equation*}
So  there must exist some $t$ and $\tilde{q}_t$ such that
\begin{equation}\label{eq:inside}
    \delta [\pi_L(q_L^*)-\pi_H(q_H^*)] \geq  (1-\delta) [T(t)- t - 1]  \left[\pi_L(q_L^*) - \pi_H(\tilde{q}_t)\right],
\end{equation}
and
\begin{equation}\label{eq:rentapproximate}
        (1-\delta)\sum_{\tau = t+1}^{T(t)-1}\delta^{\tau-t}\Delta C (\tilde{q}_t) = \Delta C (q_H^*) + O(\epsilon).
\end{equation}

Now let us send $\epsilon$ to $0$, which also pushes $\delta$ to $1$ (remember $\delta > 1-\epsilon$). Taking a subsequence if necessary, let $\tilde{q} := \lim_{\epsilon\to 0}\ \tilde{q}_t$. By \eqref{eq:rentapproximate}, we have
\begin{equation*}
    \lim_{\epsilon\to 0}\  \delta^{T(t)-t-1} = \frac{\Delta C (\tilde{q}) - \Delta C (q_H^*)}{\Delta C (\tilde{q})},
\end{equation*}
which implies
\begin{equation*}
    \lim_{\epsilon\to 0}\  -\ln \delta \left[T(t)-t-1\right] = \ln \left( \frac{\Delta C (\tilde{q})}{\Delta C (\tilde{q}) - \Delta C (q_H^*)}\right).
\end{equation*}
Note also that $\lim_{\epsilon\to 0}\frac{1-\delta}{-\ln\delta} =1$, so taking the limit on both sides on \eqref{eq:inside} renders
\begin{equation*}
     \pi_L(q_L^*)-\pi_H(q_H^*) \geq  \ln \left( \frac{\Delta C (\tilde{q})}{\Delta C (\tilde{q}) - \Delta C (q_H^*)}\right)  \left[\pi_L(q_L^*) - \pi_H(\tilde{q})\right].
\end{equation*}
Hence we have found a $\tilde{q}\in [q_L^*,\bar{q}]$ that violates \eqref{eq:nsconditionnonrevealing}, a contradiction.

Lastly, we argue that no optimal mechanism can be revealing when $\delta$ is sufficiently large. Suppose (toward a contradiction) that there always is a revealing mechanism that generates a profit of $\pi_H(q_H^*)$ as $\delta$ gets close to $1$. When proving Theorem \ref{thm:structure}, the presumption that $\Pi_0>\pi_H(q_H^*)$ is only used to show that $r_0>0$ in \textit{every} optimal mechanism. Even if $\Pi_0=\pi_H(q_H^*)$, we can still find \textit{an} optimal mechanism with $r_0>0$,\footnote{This mechanism can be constructed by skipping the first several periods without information revelation.} thus this optimal mechanism satisfies all properties in Theorem \ref{thm:structure} and Lemmas \ref{lem:frontloading} through \ref{lem:separatingquantity}. But then, one can apply exactly the same argument prior to this paragraph to reach the same contradiction, which establishes uniqueness in this case.
\end{proof}

\begin{proof}[\bf Proof of Part 2 of Theorem \ref{thm:nsc}]
Suppose that \eqref{eq:nsconditionrevealing} holds for some $\tilde{q}_L\in [q_L^*,\bar{q}]$. We will construct an implementable allocation with information revelation that generates a profit strictly greater than $\pi_H(q_H^*)$.

Let $\tilde{q}\in [q_L^*,\bar{q}]$ be such that
\begin{align}\label{eq:nsconditionrevealingslack}
        \pi_L(q_L^*) - \pi_H(q_H^*) > \ln \left(\frac{\Delta C(\tilde{q})}{\Delta C(\tilde{q})-\Delta C({q}_H^*)}\right) \left[\pi_L(q_L^*)-\pi_H(\tilde{q})\right]+\epsilon,
\end{align}
for some $\epsilon>0$. Consider the following allocation: For every $t$,
\begin{align*}
    q_t^P &= q_H^*,\\
    x_t^P &= C_H(q_H^*),\\
    q_\tau^S(t) &= \begin{cases}
        \tilde{q}, &\text{ if }t\leq \tau\leq t+T\\
        q_L^*, &\text{ if } \tau > t+T
    \end{cases}\\
    x_\tau^S(t) &= \begin{cases}
        C_H(\tilde{q}), &\text{ if }t\leq \tau < t+T\\
        C_L(\tilde{q}) + \beta \Delta C(\tilde{q}), &\text{ if } \tau = t+T\\
        C_L(q_L^*), &\text{ if } \tau > t+T
    \end{cases}\\
    r_t &= r_0 \rho ^{t},
\end{align*}
where $\rho$ is less than but close to $1$, $r_0= (1-\rho)\alpha_0$, $T$ is the smallest integer that satisfies $(1-\delta)\sum_{\tau=t}^{t+T}\delta^{\tau-t}\Delta C(\tilde{q}) \geq \Delta C(q_H^*)$, and $\beta\in (0,1]$ is such that 
\begin{equation}\label{eq:rentapproximate2}
    (1-\delta)\sum_{\tau=t}^{t+T-1}\delta^{\tau-t}\Delta C(\tilde{q}) + (1-\delta)\delta^T\beta\Delta C(\tilde{q}) = \Delta C(q_H^*).
\end{equation}

It is easy to verify that the constructed allocation satisfies \eqref{eq:IR-H}, \eqref{eq:IR-L}, \eqref{eq:IC-H}, \eqref{eq:IC-L}. We now show that when $\delta$ is sufficiently large, we can always find some $\rho$ that keeps \eqref{eq:nonreneging} satisfied, so that the constructed allocation is implementable. It suffices to ensure that in each period $t$, the principal weakly prefers waiting one more period to renege than reneging right away. Adapting condition \eqref{eq:NRwaiting} to the constructed allocation, we need
\begin{align*}
        &r_t[(1-\delta)\pi_H(\tilde{q})+\delta \pi_L(q_L^*)-\pi_H(q_H^*)] \\
        \geq 
        &(1-\delta)\sum_{\tau=t-T}^{t-1}r_t\rho^{\tau-t} [\pi_L(q_L^*)- \pi_{H}(\tilde{q})] + (1-\delta)r_t\rho^{-T}(1-\beta)\Delta C(\tilde{q}),
    \end{align*}
or equivalently, canceling $r_t$ on both sides,
\begin{align}\label{eq:NRwaiting2}
        &(1-\delta)[\pi_H(\tilde{q})-\pi_H(q_H^*)]+\delta [\pi_L(q_L^*)-\pi_H(q_H^*)] \nonumber\\
        \geq 
        &(1-\delta)\sum_{\tau=t-T}^{t-1}\rho^{\tau-t} [\pi_L(q_L^*)- \pi_{H}(\tilde{q})] + (1-\delta)\rho^{-T}(1-\beta)\Delta C(\tilde{q}).
    \end{align}
When $\delta$ and $\rho$ both go to $1$, the LHS of \eqref{eq:NRwaiting2} converges to $\pi_L(q_L^*)-\pi_H(q_H^*)$, and the RHS of \eqref{eq:NRwaiting2} converges to $\ln \left(\frac{\Delta C(\tilde{q})}{\Delta C(\tilde{q})-\Delta C({q}_H^*)}\right) \left[\pi_L(q_L^*)-\pi_H(\tilde{q})\right]$.\footnote{Note that $T$ is determined by \eqref{eq:rentapproximate2}, so it is independent of $\rho$. The limit of the RHS of \eqref{eq:NRwaiting2} is taken by first sending $\rho$ to $1$ for each given $\delta$, rendering $(1-\delta)T[\pi_L(q_L^*)- \pi_{H}(\tilde{q})]+(1-\delta)(1-\beta)\Delta C(\tilde{q})$, and then sending $\delta$ to $1$, rendering $\ln \left(\frac{\Delta C(\tilde{q})}{\Delta C(\tilde{q})-\Delta C({q}_H^*)}\right) \left[\pi_L(q_L^*)-\pi_H(\tilde{q})\right]$.} But because $\tilde{q}$ satisfies \eqref{eq:nsconditionrevealingslack}, this means that for each sufficiently large $\delta$, one can find $\rho$ that makes \eqref{eq:NRwaiting2} hold, thus \eqref{eq:nonreneging} is satisfied.

We conclude the proof by showing that the principal's profit from the constructed allocation is strictly greater than $\pi_H(q_H^*)$. Because \eqref{eq:nonreneging} is satisfied, the principal's total time-$0$ profit is weakly greater than that from reneging at $t=1$, that is:
\begin{equation*}
    \Pi_0 \geq (1-r_0)\pi_H(q_H^*) +  r_0 \left[(1-\delta)\pi_H(\bar{q}) + \delta \pi_L(q_L^*)\right]> \pi_H(q_H^*),
\end{equation*}
where the last inequality holds as long as $\delta > \frac{\pi_H(q_H^*) - \pi_H(\bar{q})}{\pi_L(q_L^*) - \pi_H(\bar{q})}$. Therefore, every optimal mechanism is revealing whenever $\delta$ is sufficiently close to $1$.
\end{proof}

\begin{proof}[\bf Proof of Proposition \ref{prop:typedistance}]
Let us first ignore the upper bound $\bar{q}$ on quantity. We first show that there exists $\theta_1$ such that if $\theta_L<\theta_1$ then \eqref{eq:nsconditionrevealing} is satisfied by setting $\tilde{q}_L = q_L^*$. Notice that $q_L^*$ goes to $\infty$ as $\theta_L$ goes to $0$. So we have
\begin{align*}
    &\lim_{\theta_L\to 0 }\ \pi_L(q_L^*) - \pi_H(q_H^*) - \ln \left(\frac{ q_L^*}{q_L^*-q_H^*}\right) \left[\pi_L(q_L^*)-\pi_H({q}_L^*)\right] \\
    = &\lim_{q_L^*\to \infty }\ \pi_L(q_L^*) - \pi_H(q_H^*) - \ln \left(\frac{ q_L^*}{q_L^*-q_H^*}\right) \theta_Hq_L^*\\
    =& \lim_{q_L^*\to \infty }\ \pi_L(q_L^*) - \pi_H(q_H^*) - \theta_Hq_H^*\\
    =& \lim_{q_L^*\to \infty }\ \pi_L(q_L^*) - v(q_H^*)\\
    =& \lim_{q_L^*\to \infty }\ v(q_L^*) - v(q_H^*)\\
    >& 0.
\end{align*}
Therefore, there exists $\theta_1$ such that if $\theta_L<\theta_1$, then \eqref{eq:nsconditionrevealing} is satisfied by setting $\tilde{q}_L = q_L^*$. Fix any such $\theta_1$. Let us assume that $\bar{q} > q_L^*(\theta_1)$ and let $\underline{\theta}$ be such that $q_L^*(\underline{\theta}) = \bar{q}$. The previous analysis implies that for every $\theta_L\in (\underline{\theta},\theta_1)$, \eqref{eq:nsconditionrevealing} is satisfied by some $\tilde{q}_L \in [q_L^*,\bar{q}]$.

Next we show that there exists $\theta_2$ such that if $\theta_L>\theta_2$ then \eqref{eq:nsconditionnonrevealing} is satisfied by all $\tilde{q}_L\in [q_L^*,\bar{q}]$. Assume toward a contradiction that one can find a sequence $\theta_L^n\to \theta_H$ and a sequence $\tilde{q}_L^n\in [q_L^{*n},\bar{q}]$ such that
\begin{equation*}
    \pi_L(q_L^{*n}) - \pi_H(q_H^*) \geq \ln \left(\frac{\tilde{q}_L^n}{\tilde{q}_L^n-{q}_H^*}\right) \left[\pi_L(q_L^{*n})-\pi_H(\tilde{q}_L^n)\right],\ \forall n.
\end{equation*}
Because $q_L^n\to q_H^*$, the LHS of the above inequality converges to $0$. And because $\ln\left(\frac{\tilde{q}_L^n}{\tilde{q}_L^n-{q}_H^*}\right)$ is bounded below by $\ln\left(\frac{\bar{q}}{\bar{q}-{q}_H^*}\right)$, for the inequality to hold we need $\tilde{q}_L^n\to q_H^*$. This implies that 
$$\lim_n\ \ln\left(\frac{\tilde{q}_L^n}{\tilde{q}_L^n-{q}_H^*}\right) = \infty.$$
On the other hand, because $\pi_L(q_L^{*n}) - \pi_H(q_H^*)<\pi_L(q_L^{*n})-\pi_H(\tilde{q}_L^n)$ for every $n$, we must have $$\limsup_n\ \ln\left(\frac{\tilde{q}_L^n}{\tilde{q}_L^n-{q}_H^*}\right)\leq 1,$$
a contradiction.

Define 
\begin{equation*}
    \bar{\theta} := \inf\left\{\theta\ |\ \text{if $\theta_L>\theta$, \eqref{eq:nsconditionnonrevealing} is satisfied by all } \tilde{q}_L\in [q_L^*,\bar{q}]\right\}.
\end{equation*}
The previous analysis implies that $
\bar{\theta}$ is well-defined and satisfies $\theta_1\leq \bar{\theta}\leq \theta_2$. By Part 1 of Theorem \ref{thm:nsc}, if $\theta_L>\bar{\theta}$, then the optimal mechanism is non-revealing whenever $\delta$ is sufficiently close to $1$. 

Now fix any $\theta_L<\bar{\theta}$. Define
\begin{equation*}
    \tilde{q}(\theta) : = \argmin_{\tilde{q}\in [q_L^*(\theta),\bar{q}]}\ \ln \left(\frac{\tilde{q}}{\tilde{q}-{q}_H^*}\right) \left[\pi_L(q_L^*({\theta}))-\pi_H(\tilde{q})\right].
\end{equation*}
Let us argue that $\tilde{q}(\theta) > q_L^*(\theta)$ for every $\theta$. Note that
\begin{align*}
    &\frac{d}{d\tilde{q}} \left.\ln \left(\frac{\tilde{q}}{\tilde{q}-{q}_H^*}\right) \left[\pi_L(q_L^*({\theta}))-\pi_H(\tilde{q})\right]\right|_{\tilde{q}=q_L^*(\theta)} \\
    = &\left(\frac{1}{q_L^*(\theta)} - \frac{1}{q_L^*(\theta)-q_H^*}\right)(\theta_H-\theta_L)q_L^*(\theta) + \ln \left(\frac{q_L^*(\theta)}{q_L^*(\theta)-{q}_H^*}\right)(\theta_H-\theta_L)\\
    = &\left[\ln \left(\frac{q_L^*(\theta)}{q_L^*(\theta)-{q}_H^*}\right) - \frac{q_H^*}{q_L^*(\theta)-q_H^*}\right](\theta_H-\theta_L)\\
    < & 0
\end{align*}
where the first equality follows from $\pi_H(q_L^*(\theta)) = \pi_L(q_L^*(\theta)) - (\theta_H-\theta_L)q_L^*(\theta)$ and $\pi_L'(q_L^*(\theta)) = 0$, and the inequality follows because $\ln(1+x)<x$ for all $x>0$. Hence $\tilde{q} = q_L^*(\theta)$ cannot be the minimizer of the objective function. Next define
\begin{equation*}
    G(\theta) : = \pi_L(q_L^*({\theta})) - \pi_H(q_H^*) - \ln \left(\frac{\tilde{q}({\theta})}{\tilde{q}({\theta})-{q}_H^*}\right) \left[\pi_L(q_L^*({\theta}))-\pi_H(\tilde{q}({\theta}))\right].
\end{equation*}
Note that
\begin{align*}
    G'(\theta) &= -\left[1-\ln \left(\frac{\tilde{q}({\theta})}{\tilde{q}({\theta})-{q}_H^*}\right)\right]q_L^*(\theta) - \frac{d}{d\tilde{q}}\left[\ln \left(\frac{\tilde{q}}{\tilde{q}-{q}_H^*}\right) \left[\pi_L(q_L^*({\theta}))-\pi_H(\tilde{q})\right]\right]\frac{d\tilde{q}(\theta)}{d\theta}\\
    &= -\left[1-\ln \left(\frac{\tilde{q}({\theta})}{\tilde{q}({\theta})-{q}_H^*}\right)\right]q_L^*(\theta),
\end{align*}
where the second line follows from the envelope theorem which is applicable because $\tilde{q}(\theta) > q_L^*(\theta)$. By definition of $\bar{\theta}$, there exists $\tilde{\theta}\in (\theta_L,\bar{\theta}]$ such that $G(\tilde{\theta}) \geq 0$. Remember $G(\tilde{\theta}) \geq 0$ implies $\ln \left(\frac{\tilde{q}(\tilde{\theta})}{\tilde{q}(\tilde{\theta})-{q}_H^*}\right) <1 $, so we have
\begin{equation*}
    G'(\tilde{\theta}) <0.
\end{equation*}
This argument can then be used to show that $G'(\theta)<0$ for all $\theta<\tilde{\theta}$, thus $G(\theta_L) >G(\tilde{\theta})\geq 0$. Because $\theta_L$ is arbitrary, applying Part 2 of Theorem \ref{thm:nsc}, we conclude that if $\theta_L<\bar{\theta}$, then every optimal mechanism is revealing whenever $\delta$ is sufficiently close to $1$.
\end{proof}

\section{Continuation Equilibrium After the Principal Reneges}\label{app:B}
Recall that by Assumption \ref{ass:1}, $\alpha _{0}\pi _{L}(q_{L}^{\ast })<\pi_{H}(q_{H}^{\ast }).$
Hence we only need to characterize worst payoffs for cases in
which the probability that an agent is the low-cost type conditional on having not
revealed at $t,$ $\xi _{t},$ satisfies%
\begin{equation}\label{eq:enoughhightype}
\xi_{t}\pi _{L}(q_{L}^{\ast })<\pi_{H}(q_{H}^{\ast }). 
\end{equation}

We also assume that the parties are sufficiently patient, so that%
\begin{equation}\label{eq:largedelta}
\frac{\delta }{1-\delta }\Delta C(q_{H}^{\ast })>\Delta C(\bar{q}).
\end{equation}%

\subsection{Extensive Form Game}
Each period is divided into two stages: First, the principal offers an individualized menu with at most $K<\infty$ contracts to each agent $i$: $m_{t}^{i}.$ Second, if agent $i$ chooses
contract $(x_{t}^{i},q_{t}^{i})$ from the menu $m_{t}^{i}$, production
and payments unfold accordingly; otherwise, the agent leaves the
relationship and obtains zero in every future period.

\subsubsection*{Payoff Lower Bound}
Take $\varepsilon >0$ and notice that if the principal offers a
menu containing a single contract $(q_{H}^{\ast },C(q_{H}^{\ast
})+\varepsilon )$ for each agent with an unknown type and a contract $(q_{L}^{\ast },C(q_{L}^{\ast })+\varepsilon )$ for each agent who
has revealed by $t$, then the principal obtains payoff $\varepsilon $ below the bound on the RHS of \eqref{eq:nonreneging}. Indeed, each (individualized) contract
above gives strictly positive payoffs to the agent while rejecting it gives
zero. Hence every agent has to accept it. We conclude that the payoff on the RHS of \eqref{eq:nonreneging} is a lower bound on the principal's payoff in any continuation equilibrium.

\subsection{Worse Equilibrium Construction}
\subsubsection*{Principal's Strategy}
Offer the singleton menu $\{(q_{L}^{\ast },C_{L}\left(
q_{L}^{\ast }\right) )\}$ to every agent who has revealed his type and the singleton menu  $\{(q_{H}^{\ast },C_{H}\left(q_{H}^{\ast }\right) )\}$ to every agent who has not revealed his type. Do this in every future period.

\subsubsection*{Agent's Strategy}
Agents who have revealed their type accept the contract that
yields the largest payoff in the menu provided that it is positive (choose the one with the smallest index in case of a tie). Reject every contract if all lead to negative stage-game payoffs.

For agents who have not revealed their type yet, there are two cases to consider:
\begin{itemize}
    \item \textbf{Case 1: All contracts in the menu yield a negative payoff to the high
type.}

A high type rejects all contracts in the menu. A low type also rejects all
contracts if the most profitable contract yields a negative payoff; otherwise, he selects the the most profitable contract (chooses the one with the smallest index in case of a tie).

\item \textbf{Case 2: At least one contract in the menu yields a positive payoff to
the high type.}

A high type selects the the most profitable contract (chooses the one with the smallest index in case of a tie). A low type imitates the high type.
\end{itemize}

\subsubsection*{Belief Updating}
We only need to specify belief updating in Case 2 above when a contract that differs from the
high type's choice prescription is chosen. In this case, the principal
attributes probability one to the agent being low type.

\subsection{Equilibrium Verification}
\subsubsection*{High-Type's Optimality}
Since his continuation payoff is always zero and the strategy above maximizes his stage-game payoff, he does not have a profitable deviation.

\subsubsection*{Low-Type's Optimality}
For anyone who revealed in the past, since his continuation
payoff is always zero and the strategy above maximizes his stage-game payoff, he does not have a profitable deviation.

For anyone who has not revealed yet, the only nontrivial
case is when a menu contains a more profitable contract in the stage
game that differs from the contract chosen by the high type. We must verify
that he cannot benefit from selecting this contract. The short-term gains from selecting this
contract are at most $\left( 1-\delta \right) \Delta C(\bar{q}).$ On the other
hand, the principal learns his type after this choice and hence this choice
leads to a loss of continuation payoff of $\delta \Delta C(q_{H}^{\ast }).$
By \eqref{eq:largedelta},  this is not profitable.

\subsubsection*{Principal's Optimality}
The only nontrivial deviation to assess is the attempt to separate types among those who are still in the pool. Notice that the strategy profile above implies that the only way to
screen the agents is by making an offer that leads to rejection by the high
type. However, this is not profitable because of \eqref{eq:enoughhightype}.
\end{document}